\documentclass[twocolumn,showpacs,preprintnumbers,amsmath,amssymb,superscriptaddress]{revtex4-1}

\usepackage{graphicx,color,amssymb,amsthm,hyperref}% Include figure files
\usepackage{dcolumn}% Align table columns on decimal point
\usepackage{bm}% bold math
\usepackage{bbm}

\newcommand{\bra}[1]{\langle #1|}
\newcommand{\ket}[1]{|#1\rangle}
\newcommand{\braket}[2]{\langle #1|#2\rangle}

\usepackage{multirow,array,rotating}
\newcolumntype{M}{>{$\vcenter\bgroup\hbox\bgroup}c<{\egroup\egroup$}}
\newcommand{\ve}{\boldsymbol}
\newcommand{\abs}[1]{\left\lvert#1\right\rvert}
\newcommand{\elip}{\mathcal{E}}
\newcommand{\A}{\mathcal{A}}
\newcommand{\nrm}[1]{\left\lVert{#1}\right\rVert}
\newcommand{\can}{\widetilde}
\usepackage[T1]{fontenc}

\def\>{\rangle}
\def\<{\langle}
\def\E{ {\cal E} }

\def\A{ {\cal A} }

\def\I{ \mathbbm{1} }

\DeclareMathOperator{\rank}{rank}

\DeclareMathOperator{\tr}{tr}
\DeclareMathOperator{\diag}{diag}

\def\rank{ \mbox{rank} }
\def\diag{ \mbox{diag} }
\def\range{ \mbox{range} }

\newtheorem{lem}{Lemma}

\begin{document}

%\preprint{APS/123-QED}

\title{Quantum Steering Ellipsoids}% Force line breaks with \\

\author{Sania Jevtic}
\affiliation{Mathematical Sciences, John Crank 501, Brunel University, Uxbridge UB8 3PH, United Kingdom}
 \affiliation{Controlled Quantum Dynamics Theory, Department of Physics, Imperial College London, London SW7 2AZ, United Kingdom}
\author{Matthew Pusey}
 \affiliation{Perimeter Institute for Theoretical Physics, 31 Caroline Street North, Waterloo, ON N2L 2Y5, Canada}
 \affiliation{Controlled Quantum Dynamics Theory, Department of Physics, Imperial College London, London SW7 2AZ, United Kingdom}
\author{David Jennings}
\author{Terry Rudolph}
 \affiliation{Controlled Quantum Dynamics Theory, Department of Physics, Imperial College London, London SW7 2AZ, United Kingdom}

\date{\today}

\begin{abstract}
The quantum steering ellipsoid of a two-qubit state is the set of Bloch vectors that Bob can collapse Alice's qubit to, considering all possible measurements on his qubit. We provide an elementary construction of the ellipsoid for arbitrary states, calculate its volume, and explain how this geometric representation can be made faithful. The representation provides a range of new results, and uncovers new features, such as the existence of ``incomplete steering'' in separable states. We show that entanglement can be analysed in terms of three geometric features of the ellipsoid, and prove that a state is separable if and only if it obeys a ``nested tetrahedron'' condition.
\end{abstract}

\maketitle

The Bloch sphere provides a simple representation for the state space of the most primitive quantum unit - the qubit - resulting in geometric intuitions that are invaluable in countless fundamental information-processing scenarios.
The two-qubit system likewise constitutes the primitive unit for bipartite quantum correlations. However, the two-qubit state space is 15-dimensional, with a surprising amount of structure and complexity. As such, it is challenging both to faithfully represent its states and to acquire natural intuitions for their properties  \cite{Vis2Qubits, Tstates, Geometry}.

The phenomenon of steering was first uncovered by Schr\"{o}dinger \cite{Schrodinger} (and subsequently rediscovered by others \cite{Gisin, HJW, Rob-Terry}), who realised that local measurements on Bob's side of the pure state $\ket{\psi}_{AB}$ could be used to ``steer'' Alice's state into \emph{any} convex decompositions of her reduced state $\rho_A$. Hence we say that for $\ket{\psi}_{AB}$, steering is \textit{complete} within Alice's Bloch sphere. For a two-qubit mixed state $\rho$ it is known \cite{FrankPhD} that the convex set of states that Alice can be steered to is an ellipsoid $\E_A$, see Figure \ref{canonical}.

The purpose of this Letter is to show that this steering ellipsoid is the natural generalization of the Bloch sphere picture, in that it can be used to give a faithful representation of an arbitrary two-qubit state in three dimensions, and moreover that the core properties of the state and its correlations are made manifest in simple geometric terms.

 By adopting this representation we are led to a range of novel results for both separable and entangled states.

Firstly, it reveals a new feature of separable quantum states, called incomplete steering, where not all decompositions of $\rho_A$ within the steering ellipsoid $\E_A$ are accessible. More importantly, the representation reveals surprising structure in mixed state entanglement. We find that mixed state entanglement decomposes into the simple geometric components of (a) the spatial orientation of the ellipsoid, (b) its distance from the origin and (c) its size. We are also lead to the surprising \textit{nested tetrahedron condition}: a state is separable if and only if its ellipsoid fits inside a tetrahedron that itself fits inside the Bloch sphere.

The representation also provides unity and insight for a range of distinct features. The nested tetrahedron condition leads to a simple determination of the minimal number of product states in the ensemble of any separable state. We note that the ellipsoid volume is an entanglement criterion, and provide a formula for it in terms of $\det(\rho)$ and $\det(\rho^{T_B})$. Non-zero ellipsoid volume is a type of correlation intermediate between discord and entanglement.

\begin{figure}[t]
\includegraphics[width=\columnwidth]{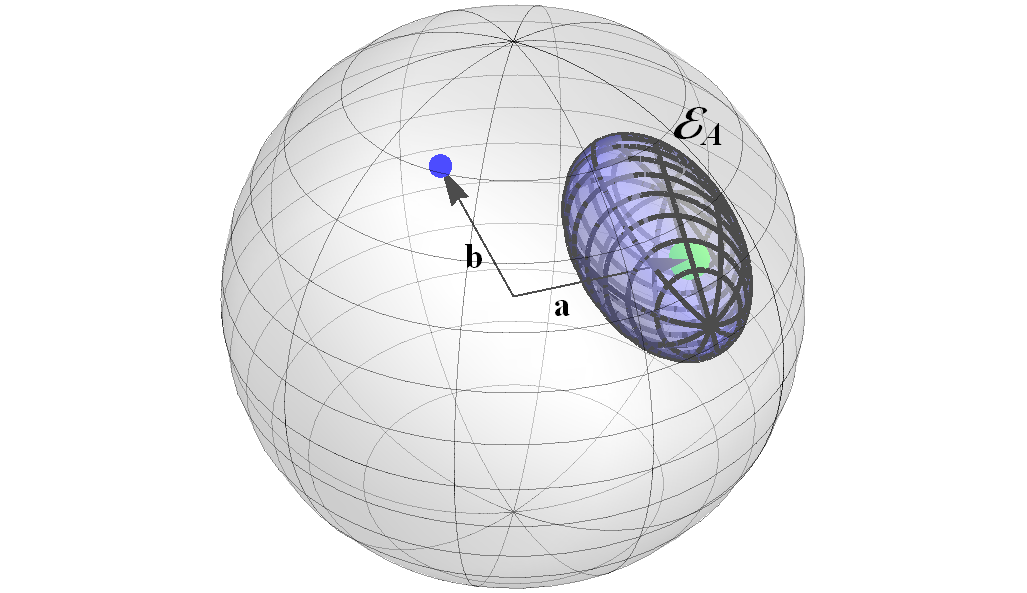}
\caption{\textbf{Ellipsoid representation of a two-qubit state.} For any two-qubit state $\rho$, the set of states to which Bob can steer Alice forms an ellipsoid $\E_A$ in Alice's Bloch sphere, containing her Bloch vector $\ve{a}$. Bob's Bloch vector $\ve b$ is also shown.} 
\label{canonical}
\end{figure}

Beyond these new  insights, we also feel that this method of compactly depicting any two-qubit state in three dimensions should be of interest to a range of researchers in both the theoretical and experimental quantum sciences.

\textit{The Pauli basis.} Let ${\sigma_\mu}= \{\I, \sigma_x,\sigma_y,\sigma_z\}$, $\mu=0, 1,2,3$ denote the ``homogeneous Pauli basis". Any single-qubit Hermitian operator $\hat E$ can be written $\hat E = \frac12\sum_{\mu=0}^3 X_\mu \sigma_\mu$, where the $X_\mu = \tr(\hat E\sigma_\mu)$ are components of the real vector $X$. Demanding that $\hat{E}\geq 0$ is equivalent to $X_0 \geq 0$ and $X_0^2 \geq \sum_{i=1}^3 X_i^2 $, and we can identify $\hat{E}$ as a POVM element. %so with $X$ viewed as a 4-vector, we can identify the set of positive operators with the forward light cone in Minkowski space.

In a similar way, any two-qubit state $\rho$ can be written in the Pauli basis as $\rho=\frac{1}{4}\sum_{\mu ,\nu =0}^{3}\Theta _{\mu \nu }\sigma _{\mu
}\otimes \sigma _{\nu }$, where $\Theta_{\mu\nu} = \tr(\rho \sigma_\mu \otimes \sigma_\nu)$ is real for all $\mu, \nu$. As a block matrix we have $\Theta=
\begin{pmatrix}
    1 & \boldsymbol{b}^T \\
    \boldsymbol{a} & T
\end{pmatrix}$, where $\boldsymbol{a},\boldsymbol{b}$ are the Bloch vectors of the reduced states $\rho_A$ and $\rho_B$ of $\rho$ respectively, and $T$ is a 3 $\times$ 3 matrix encoding the correlations \cite{Tstates}. If Bob does a POVM and obtains outcome $\hat E$, he steers Alice to the state proportional to $\tr_B(\rho (\I \otimes \hat E))$, which in the Pauli basis is given by the 4-vector $\frac12 \Theta X$, with probability $\frac12 (1+\ve b \cdot  \ve{x} )$ where $\ve x = (X_1,X_2,X_3)^T$.

The $4$-vector formalism is related to the idea of stochastic local operations and classical communication (SLOCC) \cite{Vis2Qubits}, which are operations of the form $\rho \rightarrow \rho' = S_A\otimes S_B \rho (S_A \otimes S_B)^\dag$, where $S_A, S_B$ are invertible complex matrices. The set of states attainable from $\rho$ under SLOCC is called the SLOCC orbit of $\rho$, and denoted $\mathcal{S}(\rho)$. Under this action the matrix $\Theta$ transforms as $\Theta' = \Lambda_A \Theta \Lambda_B^T$ where $\Lambda_{A(B)}$ are proper orthochronous Lorentz transformations (Appendix \ref{sec_SteerEll_Deriv}). Significant in what follows, for a SLOCC operation affecting only Bob ($\Theta' = \Theta\Lambda_B$) the set of states Alice is steered to is unaffected, since: $X$ is in the forward light cone if and only if $X' = \Lambda_B X$ is, and $\Theta' X = \Theta X'$.

Previously, in \cite{FrankPhD} a range of SLOCC techniques were employed to study entanglement and steering for two-qubit mixed states, however this approach encounters problems when applied to certain separable states and moreover is not suited to addressing the geometric features of interest. The techniques developed here follow a different line, and circumvent both of these issues.

\textit{Construction of the quantum steering ellipsoid.}
We now provide an alternative construction of the steering ellipsoid $\E_A$ to that in \cite{FrankFilter}, which applies even when $\E_A$ is degenerate.

Our construction of $\E_A$ is easiest to understand in the case when the state $\rho$ has $\boldsymbol{b}=\boldsymbol{0}$. For such a state, suppose Bob projects his qubit onto the pure state $X = \begin{pmatrix}1 \\ \boldsymbol{x}\end{pmatrix}$ with $x = 1$. Given this outcome Alice is steered to
\begin{equation}
  Y = \Theta X = \begin{pmatrix}1 & \boldsymbol{0}^T \\ \boldsymbol{a}&T\end{pmatrix}\begin{pmatrix}1 \\ \boldsymbol{x}\end{pmatrix} =  \begin{pmatrix} 1 \\ \boldsymbol{a} + T\boldsymbol{x}\end{pmatrix},
\end{equation}
which occurs with probability $\frac 12$ and where Alice's Bloch vector is now $\boldsymbol{a} + T\boldsymbol{x}$. The set of all states Alice can end up with is simply the unit sphere of possible $\boldsymbol{x}$, shrunk and rotated by $T$ and translated by $\boldsymbol{a}$, i.e.\!\! an ellipsoid centred at $\boldsymbol{a}$ with orientation and semiaxes given by the eigenvectors and eigenvalues of $T T^T$. The ellipsoid dimension is $\rank(T) = \rank(\Theta) - 1$. Points inside the ellipsoid can be reached via convex combinations of projective measurements, and conversely a POVM element is a positive operator and so can be spectrally decomposed into a mixture of projectors, thus giving a point within the ellipsoid.

Now consider a general state with $\boldsymbol{b}\neq\boldsymbol{0}$. If $b=1$ then $\rho$ is a product state in which case there is no steering and the steering ellipsoid is the single point $\ve{a}$. For the case $b < 1$, we find that the SLOCC operator $\I\otimes(2\rho_B)^{-\frac{1}{2}}$ corresponds to a Lorentz boost $L_{\boldsymbol{b}}$ by a `velocity' $\boldsymbol{b}$ that transforms $\rho_B$ to the maximally mixed state (which has $\boldsymbol{b} = \boldsymbol{0}$). We refer to this special filtered state $\widetilde{\rho}$ as the \textit{canonical state} on the SLOCC orbit $\mathcal{S}(\rho)$.
Since SLOCC operations on Bob do not affect Alice's steering ellipsoid, the parameters of an arbitrary state's steering ellipsoid are obtained by simply boosting $\Theta$ by $L_{\boldsymbol{b}}$ and reading off the ellipsoid parameters. This gives a steering ellipsoid centred at $\boldsymbol{c}_A = \frac{\boldsymbol{a}-T\boldsymbol{b}}{1-b^{2}}$, with orientation and semiaxes lengths $s_i = \sqrt{q_i}$ given (Appendix~\ref{sec_SteerEll_Deriv}) by the eigenvectors and eigenvalues $q_i$ of the ellipsoid matrix
\begin{equation}\label{QA}
Q_A =\frac{1}{1-b^2}\left( T-\boldsymbol{a}\boldsymbol{b}%
^{T}\right) \left(  \I+\frac{\boldsymbol{%
b}\boldsymbol{b}^{T}}{ 1-b^{2}}\right) \left( T^{T}-\boldsymbol{b}\boldsymbol{a}^{T}\right).
\end{equation}

To obtain $\E_B$, the ellipsoid at B, we simply perform a swap of $A$ and $B$, which corresponds to transposing $\Theta$ and sends $\boldsymbol{b}\rightarrow\boldsymbol{a}, \boldsymbol{a}\rightarrow\boldsymbol{b}, T \rightarrow T^T$. Hence $\E_A$ and $\E_B$ always have the same dimensionality, $\rank(\Theta)-1$. This completes the construction of the geometric data $(\E_A, \boldsymbol{a}, \boldsymbol{b})$ for a given state $\rho$. Next, we describe the reverse direction: obtaining $\rho$ from an ellipsoid $\E_A$ and the vectors $\boldsymbol{a}$ and $\boldsymbol{b}$.

\textit{Reconstruction of $\rho$ from geometric data.} Given $\boldsymbol{a}, \boldsymbol{b}, \E_A = (Q_A, \ve{c}_A)$, to recover $\rho$ we need $T$. In Appendix~\ref{sec_RecRho} we prove that this matrix is given by
\begin{equation}
T=\frac{1}{\gamma} \left(\gamma \ve{c}_A\ve{b}^T + \sqrt{Q_A}O + \frac{\gamma-1}{b^2}\sqrt{Q_A}O\ve{b}\ve{b}^T \right)
\label{Teq}
\end{equation}
where $O \in \mathrm{O}(3)$ satisfies $ \boldsymbol{a} = \ve{c}_A + \sqrt{Q_A}O\ve b$. This specifies $O$ up to a rotation $O'\in \mathrm{O}(3)$ such that $O'\ve b = \ve b$. The action of $O'$ can be encoded, for example, by a colouring of $\E_A$ as in \cite{Altepeter}. In this way the steering ellipsoid can be used as a faithful representation of $\rho$. $O'$ corresponds to a local unitary and/or partial transpose on Bob's system, and so is irrelevant for any correlation properties such as entanglement.

\textit{``Complete''  and ``incomplete" steering.}
The steering ellipsoid specifies which \emph{states} Bob can steer Alice to. A more subtle question is which \emph{decompositions} of Alice's reduced state he can steer to. Clearly a necessary condition is that all of the states in the decomposition must be in $\mathcal{E}_A$, surprisingly however, it turns out that this is not sufficient.

Consider some non-product two-qubit state with ellipsoids $\mathcal{E}_A$ and $\mathcal{E}_B$. The following are equivalent (Appendix~\ref{sec_compthm}):
\begin{enumerate}
  \item (\emph{Complete steering of A}) For any convex decomposition of $\boldsymbol{a}$ into vectors in $\mathcal{E}_A$ or on its surface, there exists a POVM for Bob that steers Alice to it.
  %\item Alice's Bloch vector $\boldsymbol{a}$ lies on the surface of $\mathcal{E}_A$ scaled down by $b$.
  \item The affine span of $\mathcal{E}_B$ contains the origin.
\end{enumerate}
In particular, these conditions hold for all non-degenerate ellipsoids (which includes all entangled states) as well as all states where $\boldsymbol{b} = \boldsymbol{0}$. However, complete steering is not symmetric: the state $\rho=\frac12\left( \ket{00}\bra{00} + \ket{+1}\bra{+1} \right)$ has complete steering of Alice by Bob, but not vice versa.

\begin{figure}[h!]
\begin{tabular}{|M|M|M|M|}
Dim. & Type & Alice's ellipsoid & Bob's ellipsoid \\
3 \begin{sideways}(Obese)\end{sideways}& \begin{sideways}Ent. \& Comp. \end{sideways} & \includegraphics[width=1.2in]{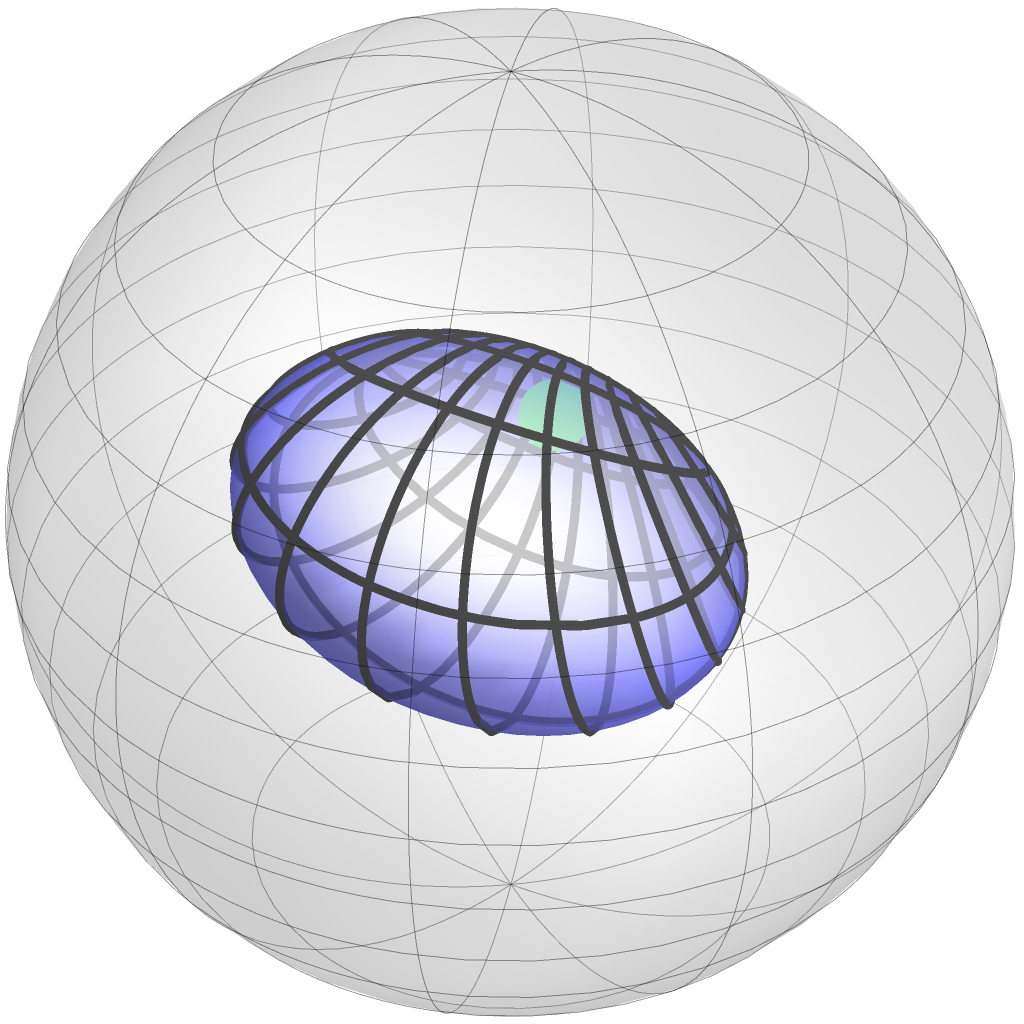} & \includegraphics[width=1.2in]{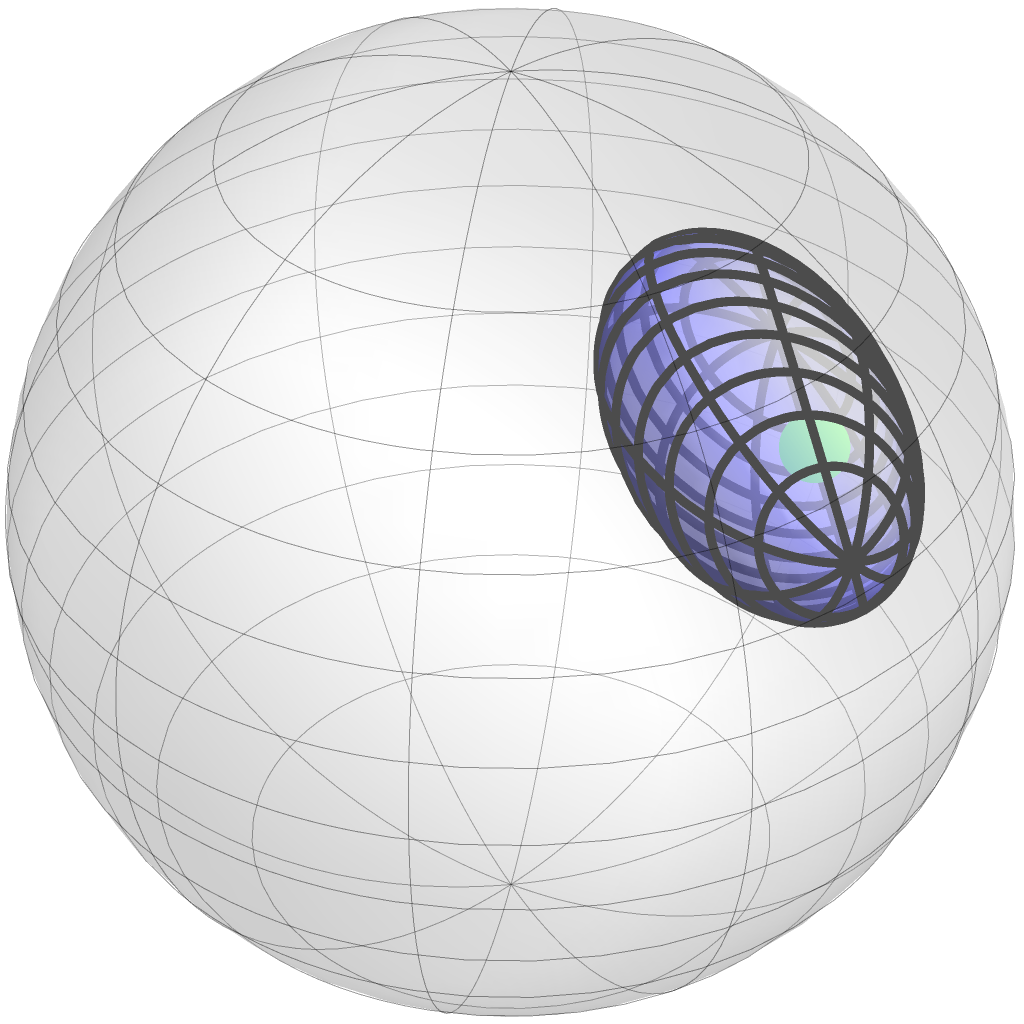}\\
3 \begin{sideways}(Obese)\end{sideways}& \begin{sideways}Sep. \& Comp.\end{sideways} & \includegraphics[width=1.2in]{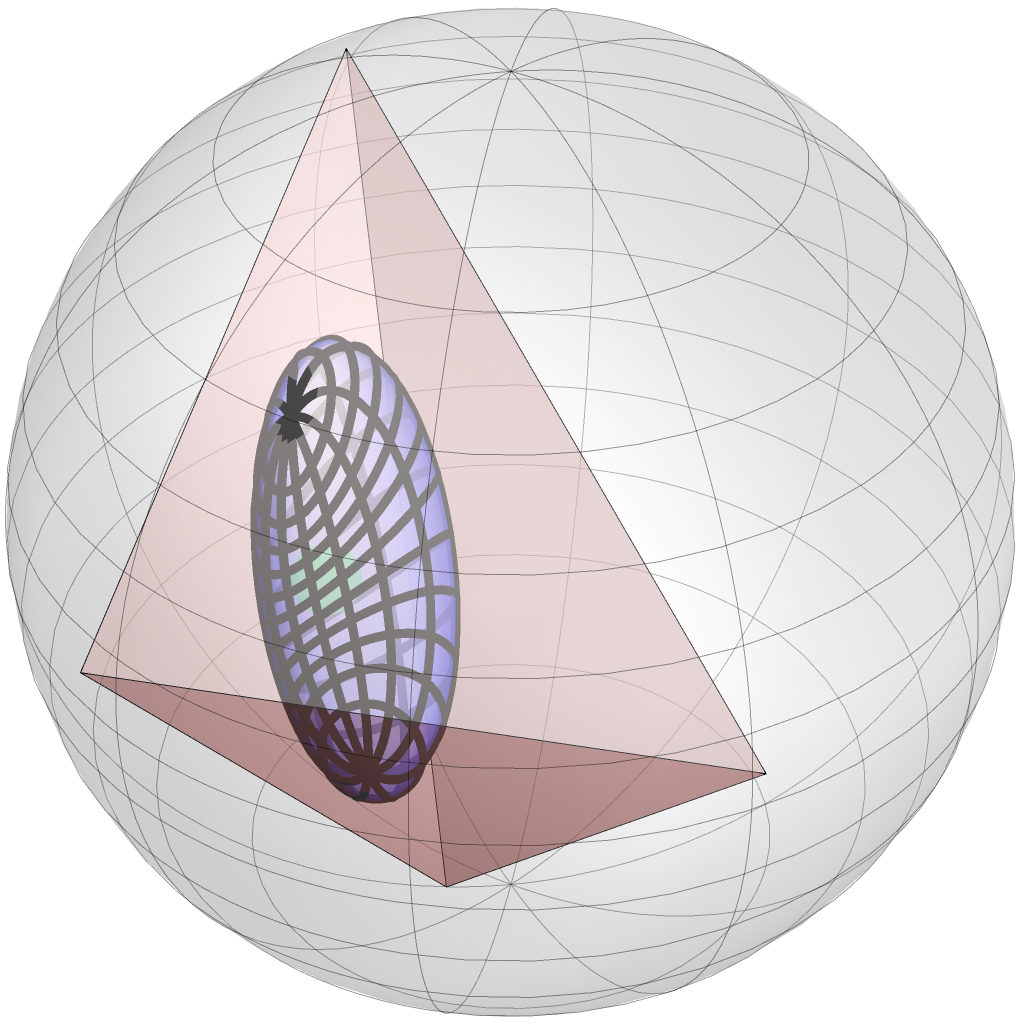} & \includegraphics[width=1.2in]{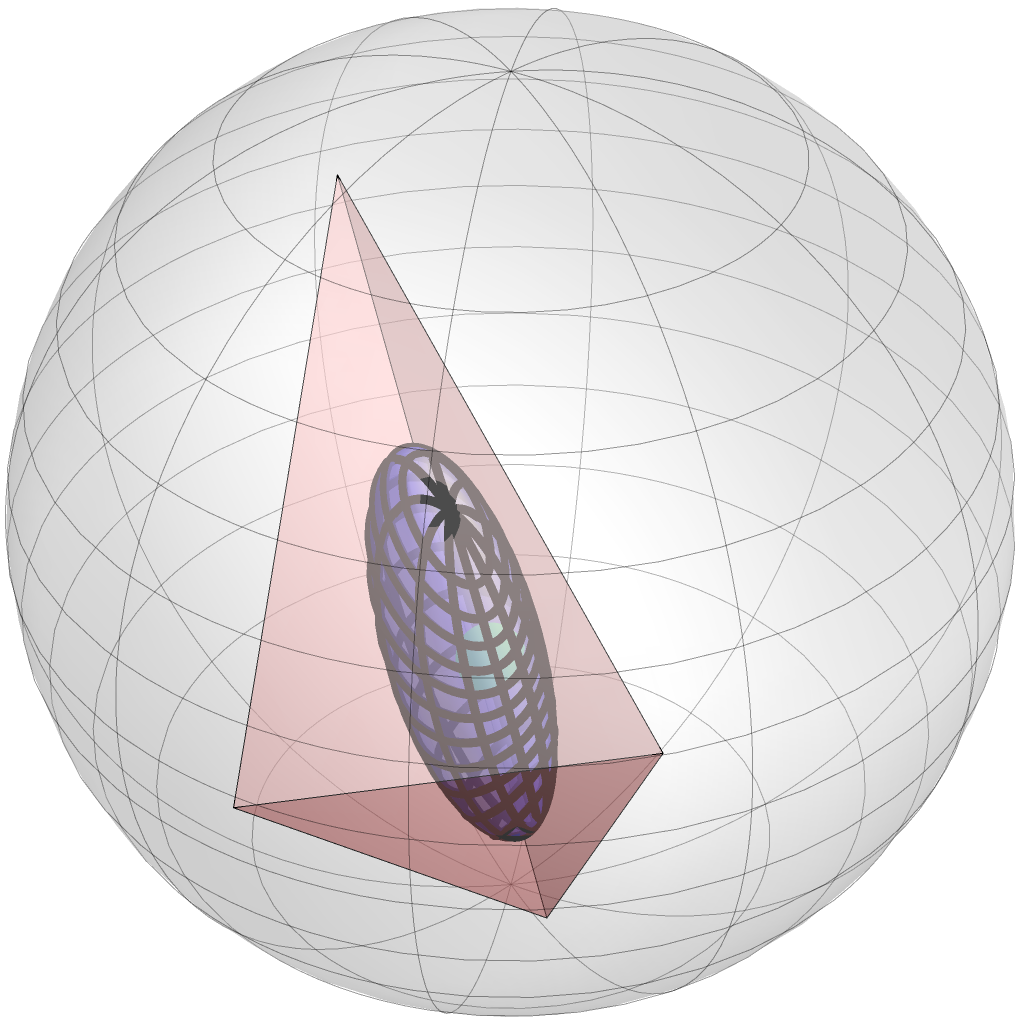}\\
2 \begin{sideways}(Pancake)\end{sideways}& \begin{sideways}Sep. \& Incomp.\end{sideways} & \includegraphics[width=1.2in]{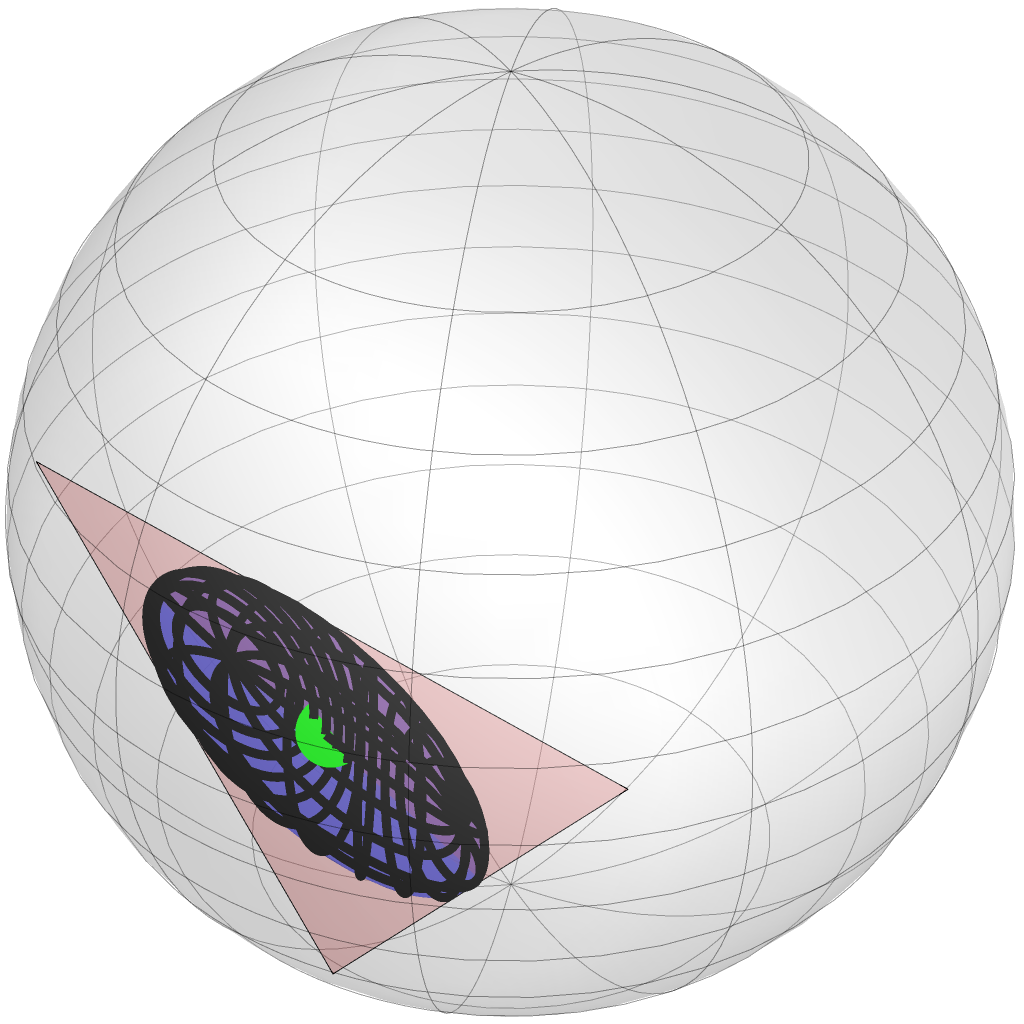} & \includegraphics[width=1.2in]{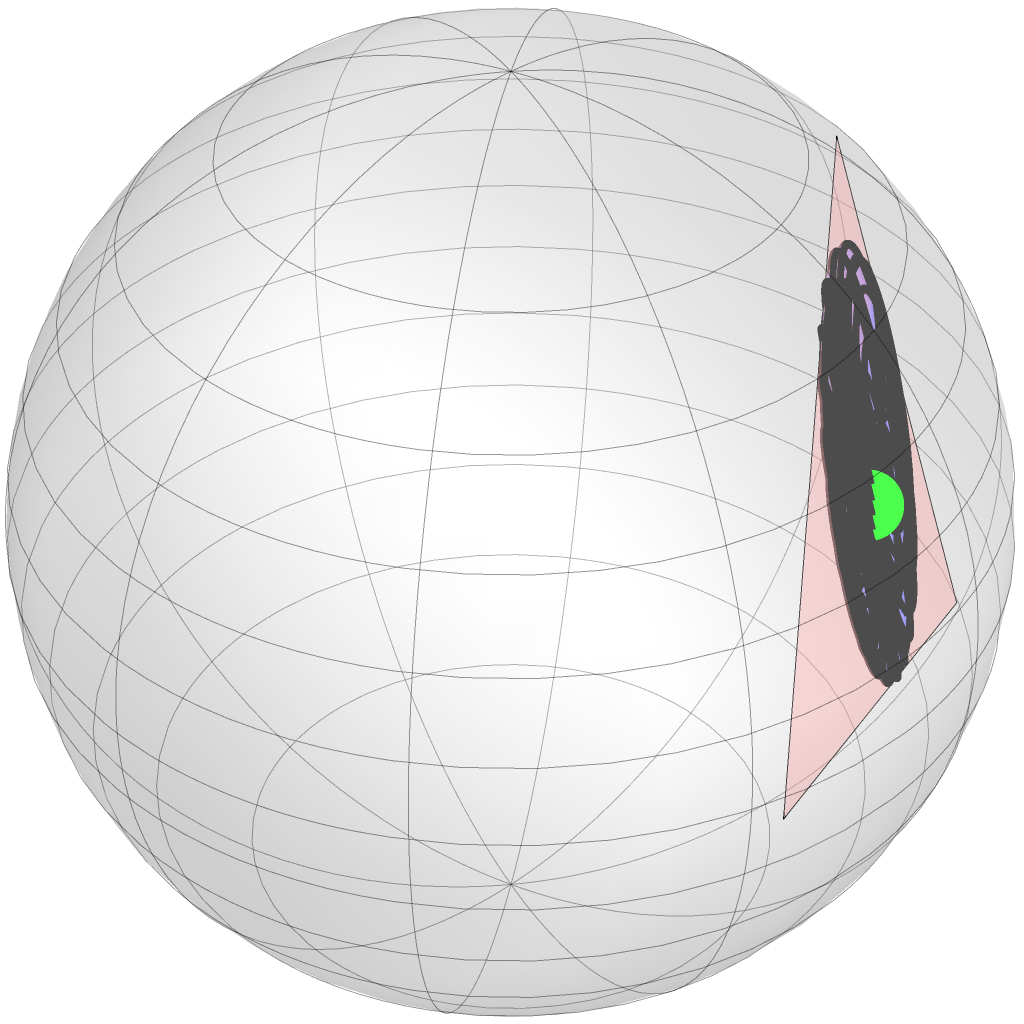}\\
2 \begin{sideways}(Pancake)\end{sideways}& \begin{sideways}Sep. \& Comp.\end{sideways} & \includegraphics[width=1.2in]{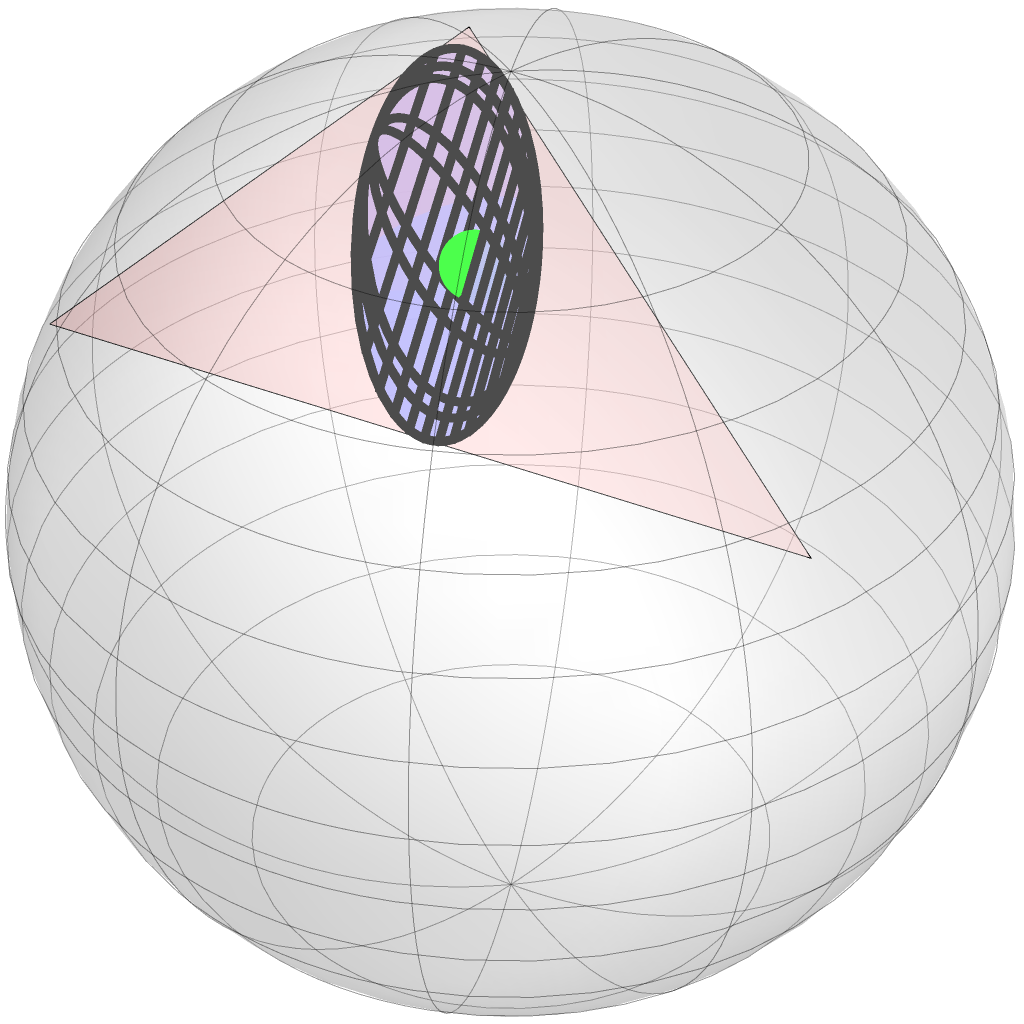} & \includegraphics[width=1.2in]{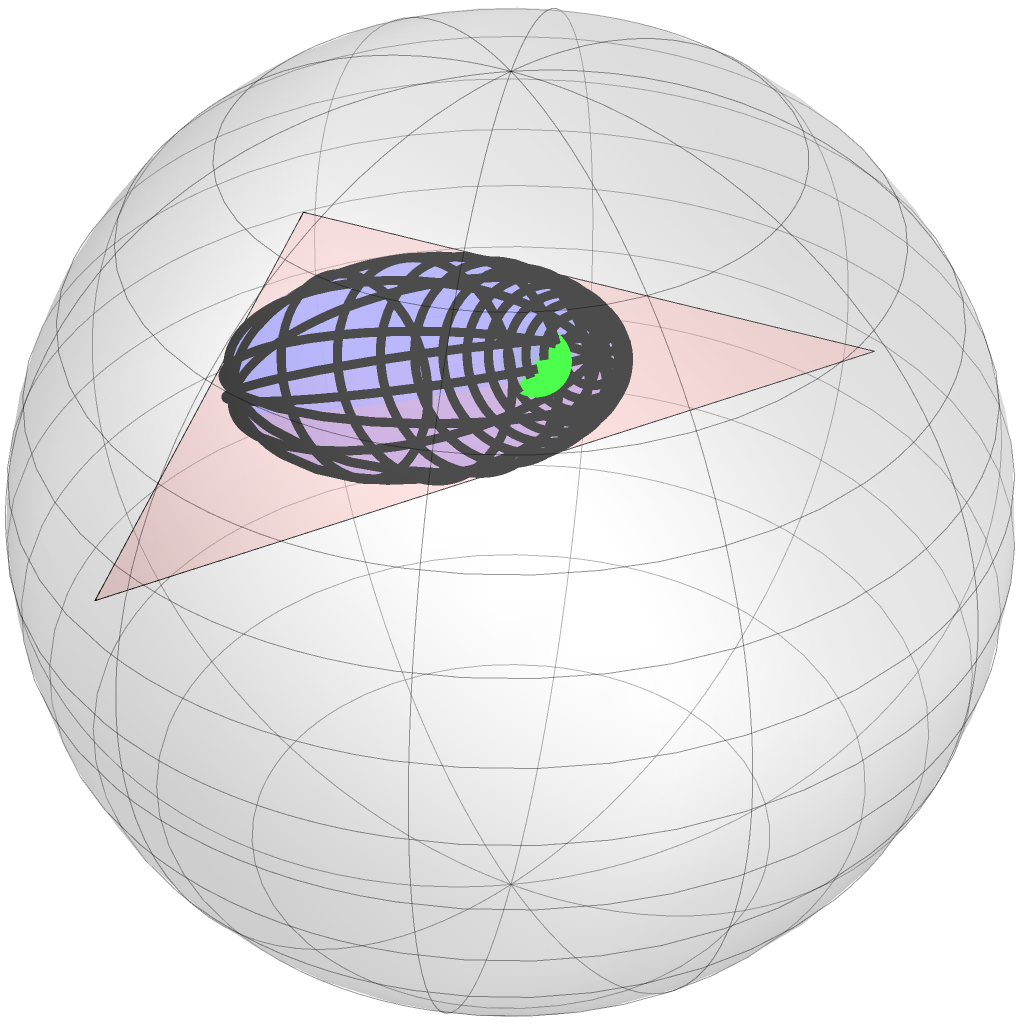}\\
1 \begin{sideways}(Needle)\end{sideways}& \begin{sideways}Sep. \& Incomp. \end{sideways} & \includegraphics[width=1.2in]{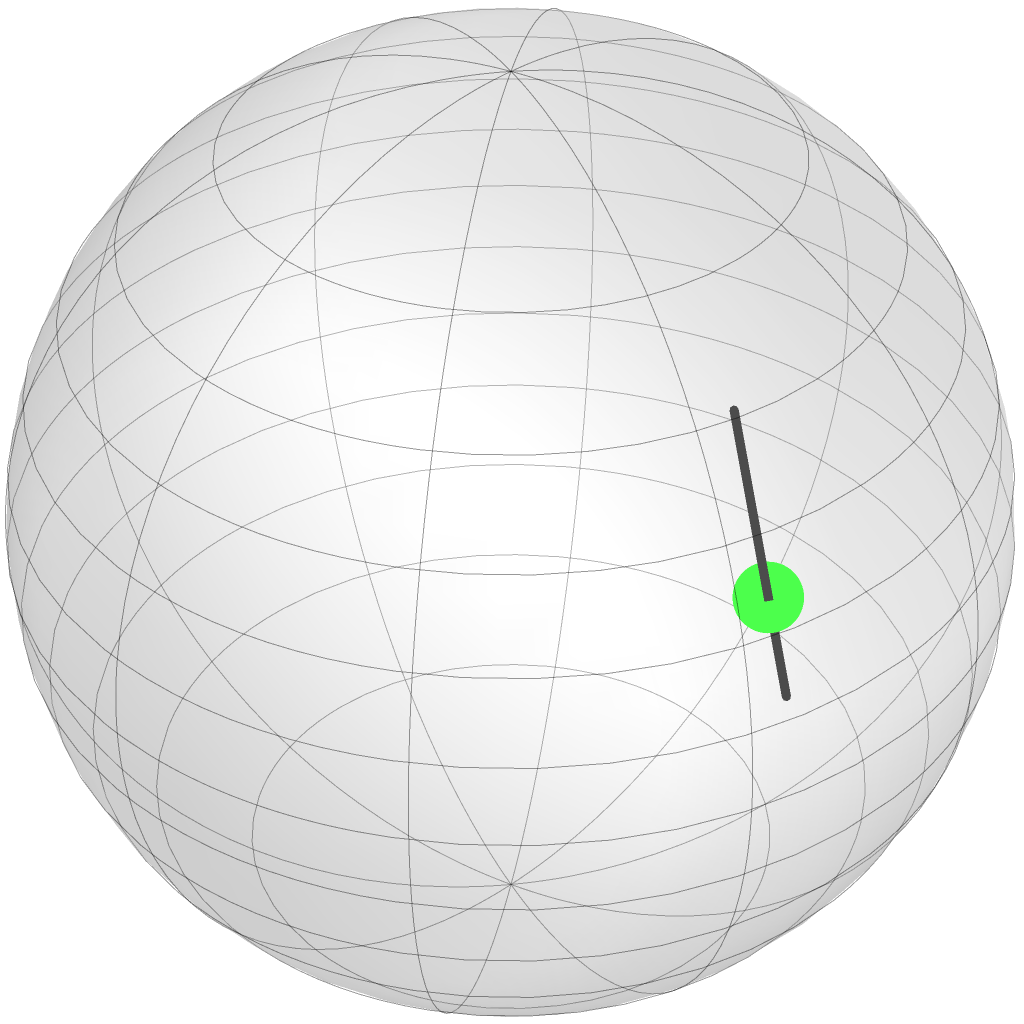} & \includegraphics[width=1.2in]{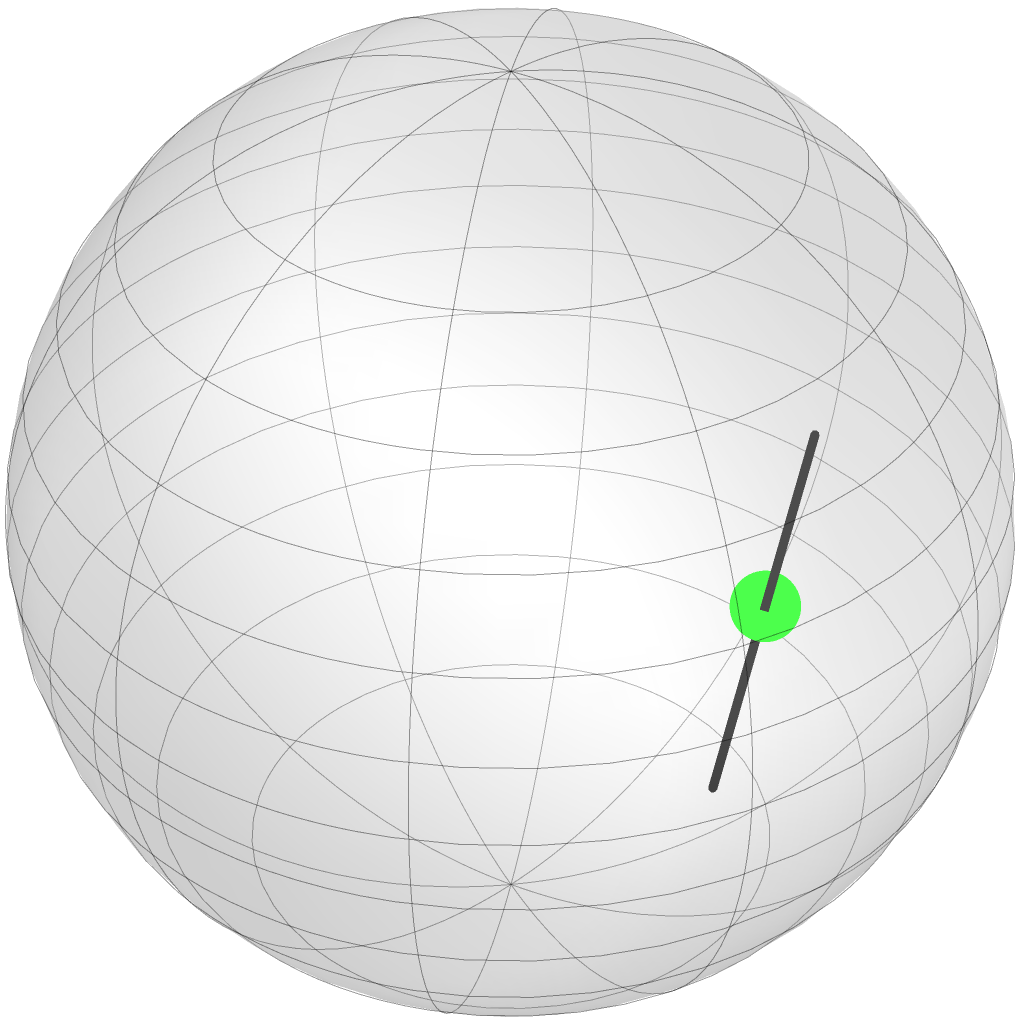}%\\
\end{tabular}
\caption{\textbf{The classes of steering ellipsoids}. Here Sep. = ``Separable", Ent. = ``Entangled" labels they type of correlation, while Incomp. = ``Incomplete" and Comp.= ``Complete", labels the type of steering. The green dot is the reduced state in the respective Bloch sphere. States with $\E_A$ being 3-dimensional have non-zero volume (or simply "obesity") and these are either entangled or separable. The state $\rho$ is separable if and only if $\E_A$ fits inside a tetrahedron inside the Bloch sphere of $A$. For separable states, the set $\E_A$ can also be 2-dimensional (a steering pancake), or 1-dimensional (a steering needle), or trivially 0-dimensional (not shown). For these cases, steering is either ``complete" if all ensemble decompositions of $\boldsymbol{a}$ in $\E_A$ are attainable (when the span of $\E_B$ contains $\frac{1}{2}\I$), otherwise the steering is ``incomplete". Zero discord occurs only for radial steering needles.}\label{table-of-ellipsoids}
\end{figure}

\textit{The three geometric contributions to entanglement.} The Peres-Horodecki criterion \cite{PeresPT,HorodeckiPT} asserts that a two-qubit state $\rho_e$ is entangled if and only if $\rho_e^{T_B}$ has a negative eigenvalue. Furthermore, it can be shown \cite{FrankFilter} that at most one eigenvalue of $\rho_e^{T_B}$ can be negative, and \cite{Sanpera} that $\rho_e^{T_B}$ is full rank for all entangled states. Hence $\det \rho_e^{T_B} < 0 $ is a necessary and sufficient for entanglement.

Suppose $\rho$ is entangled, then any state in its SLOCC orbit $\mathcal{S}(\rho)$ is also entangled \cite{FrankFilter}, including the canonical state $\can \rho \in \mathcal{S}(\rho)$. It follows that $\rho$ is entangled if and only if $\det(\can{\rho}^{T_B})<0$. However, the states $\rho$ and $\can{\rho}$ share the same $\E_A$, and so expanding $\det(\can{\rho}^{T_B}) <0$ in the geometric representation we find that $\rho$ is entangled if and only if a physical steering ellipsoid with centre $\ve c = c \ve{\hat{n}}$ and matrix $Q$ satisfies
\begin{equation}
\label{DetEntCond}
c^4-2c^2(1-\tr Q + 2{\ve{\hat n}}^T Q {\ve{\hat n}}) + h(Q) <0,
\end{equation}
where $h(Q) := 1- 8  \sqrt{\det Q} +2 \tr (Q^2) -(\tr (Q))^2 - 2\tr( Q )$, and we drop $A,B$ labels as entanglement is a ``symmetric'' relation.
This equation is manifestly invariant under global rotations, corresponding to local unitaries on the quantum state, and shows that correlations between the qubits manifest themselves in three geometric ways: (1) the distance $c$ of the ellipsoid centre from the origin, (2) the size of the ellipsoid and (3) its ``skew'', captured by the term ${\ve{\hat n}}^T Q {\ve{\hat n}}$, which reflects the alignment of the ellipsoid relative to the radial direction described by centre unit vector $ {\ve{\hat n}}$.

\textit{The nested tetrahedron condition}. The condition for entanglement given by equation \eqref{DetEntCond} provides a compact algebraic condition for non-separability and uncovers contributions from different geometric aspects. However, the representation captures the distinction between separable and non-separable states in another elegant way: \begin{quote}\textit{A two-qubit state $\rho$ is separable if and only if its steering ellipsoid $\E_A$ fits inside a tetrahedron that fits inside the Bloch sphere.} \end{quote}
To prove necessity, suppose Alice and Bob share a separable state $\rho = \sum_{i=1}^n p_i \alpha_i \otimes \beta_i$. Since we can always take $n \leq 4$ \cite{PhysRevLett.80.2245}, the Bloch vectors of the $\alpha_i$ define a (possibly degenerate) tetrahedron $\mathcal{T}$ within Alice's Bloch sphere. Bob's outcome $\hat E$ collapses Alice to $\sum_{i=1}^n \frac{ \tr(\hat E \beta_i)}{\tr(\hat E \rho_B)}p_i \alpha_i$. Hence her steered Bloch vector will be a convex combination of the Bloch vectors for the $\alpha_i$ --- in other words her steering ellipsoid is contained in $\mathcal{T}$.

We prove in Appendix~\ref{Tet_proof} that the non-trivial converse holds: any ellipsoid that fits inside a tetrahedron that itself fits inside the Bloch sphere must arise from a separable state, and thus the nested tetrahedron condition is both necessary and sufficient for separability of the state.

This key geometric insight leads to some non-trivial corollaries. For example, for any separable state $\rho$, the minimal number of product states in an ensemble decomposition $\rho = \sum_{i=1}^n p_i \alpha_i \otimes \beta_i$ is $n=\rank (\Theta)$. If $\rank (\Theta) = 1$ we have a product state, and so $n=1$, while if $\rank(\Theta) = 2$ we have that $\E_A$ is a line segment and we form a decomposition of $\rho$ using the endpoints of this segment, giving $n=2$. The case $\rank (\Theta) = 3$ is slightly more involved, but follows from any ellipse inside a tetrahedron inside the unit sphere also lies inside a triangle in the unit sphere (Appendix~\ref{Tet_proof}, \cite{Mihai}). Finally it is known \cite{PhysRevLett.80.2245} that any separable state can be written using 4 product states, which covers the case $\rank(\Theta)=4$. Combining this with the above results on complete steering provides a natural geometrical classification of two qubit states, as in Fig. \ref{table-of-ellipsoids}.

\textit{Quantum discord and ellipsoid orientation.} Quantum discord  has received much attention as a measure of the quantumness of correlations (see \cite{Discord} for details) in which zero discord for one party roughly corresponds to them possessing a non-disturbing projective measurement. Within the geometric representation it is readily seen that $\rho$ has zero discord for $A$ if and only if $\E_A$ degenerates to a radial line segment, while $\rho$ has zero discord for $B$ if and only if $\E_A$ is one-dimensional and $b = \frac{2|\boldsymbol{c}_A - \boldsymbol{a}|}{l_A}$, where $l_A$ is the length of $\E_A$ (Appendix~\ref{sec_discord}).

To illustrate the effect of the alignment of $\E_A$ on the entanglement and discord of a state, we can consider a one-parameter family of states of the form $\rho(\theta) = \frac{1}{4}\left(\I + \frac{1}{2}\sigma_z\otimes \I + \sum_{ij}  T_{ij}(\theta)  \sigma_i \otimes \sigma_j\right)$, for which the ellipsoid skew varies smoothly with $\theta$ while maintaining a constant volume for $\E_A$.  Specifically, we have that $ T(\theta)=R_y(\theta) K R^T_y(\theta)=\sqrt{Q_A(\theta)} $, and so $R_y(\theta) \in \mathrm{SO(3)}$ generates states with inequivalent correlations via rotation of the steering ellipsoid around its own centre $\boldsymbol{c}_A = (0,0,\frac{1}{2})^T$, note that this ``internal rotation" is distinct from a global rotation generated by a local unitary on the state. We choose $K = \diag(-\frac{9}{20},-\frac{3}{10},-\frac{3}{10})$ so that $\rho(\theta) \geq 0$, for all $\theta \in [0,\pi)$. This family of states illustrates opposing behavior of the discord and concurrence as a function of $\theta$, see Fig. \ref{DiscConcVol}. The entanglement favors an orientation in which the longest semiaxis is aligned (radial) with $\boldsymbol{c}_A$ at the point $\theta = \pi/2$, while discord is maximised when the short semiaxis is radial, at $\theta=0,\pi$ \footnote{Discord can be calculated analytically when $\theta  =0,\frac{\pi}{2}$, and at these points $\rho$ is an X-state \cite{ShiXstateDiscord} and this agrees with our findings.}\nocite{ShiXstateDiscord}

\begin{figure}[h!]
\includegraphics[width=3.5in]{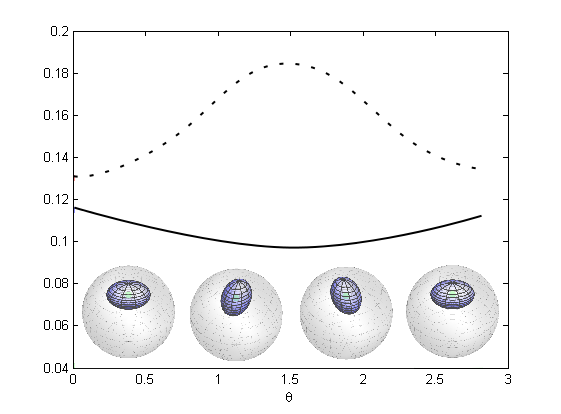}
\caption{Discord (solid) and concurrence (dotted) of the state $\rho(\theta)$ as a function of the orientation $\theta$ of the ellipsoid. Entanglement is maximized when the major axis is radial.}
\label{DiscConcVol}
\end{figure}

\textit{Volume of the ellipsoid.} The expression for the volume of $\E_A$ provides a compact and non-trivial relation between the steering properties of $\rho$ and the ranks of $\rho$ and $\rho^{T_B}$.
The volume of any ellipsoid is proportional to the product of its semiaxes $V=\frac{4\pi}{3}s_1s_2s_3$. Therefore $\mathcal{E}_A$ has volume $V_A = \frac{4\pi}{3}\sqrt{\det Q_A}$ using the ellipsoid matrix in equation \eqref{QA}, which may be rewritten as $V_A=\frac{4\pi}{3}\frac{\left\vert\mathrm{det}\Theta\right\vert}{(1-b^2)^2}$. However, it turns out (Appendix~\ref{sec_volform}) that $\mathrm{det} \Theta=16 (\mathrm{det} \rho^{T_B} - \mathrm{det} \rho)$, therefore
\begin{equation}
\label{VolumeInRho}
V_A = \frac{64\pi }{3}\frac{\left\vert \det  \rho
-\det  \rho^{T_B}  \right\vert }{\left( 1-b^{2}\right) ^{2}}.
\end{equation}
The  $\mathcal{E}_B$ volume follows from $V_A$ via the simple relation
$V_B = \frac{(1-b^2)^2}{(1-a^2)^2}V_A$.

The ellipsoid volume is a non-linear entanglement criterion. Specifically, the Werner state on the separable-entangled boundary has $\E_A$ being the maximal sphere volume $V_\star= 4 \pi/81$ inscribed inside the largest possible tetrahedron that can be inscribed inside the unit sphere \footnote{More formally, it was shown \cite{FJohn, Ball} that every convex body in $\mathbb{R}^n$ Euclidean space contains a unique maximal volume ellipsoid, and that a Euclidean ball is that maximum if and only if there exist vectors $\{\boldsymbol{x}_k\}$ in the set of contact points of the ball with the convex body, and positive numbers $\{c_k\}$ such that $\sum_k c_k \boldsymbol{x}_k = 0$ and $\sum_k c_k \boldsymbol{x}^T_k \boldsymbol{x}_k =r^2 \I$, where $r$ is the radius of the ball. Inspection of the in-sphere for the regular tetrahedron in the Bloch sphere shows that is is the case for $\{\boldsymbol{x}_k\}$ being the four intersection points.}\nocite{FJohn,Ball}. We immediately deduce that any state with $\E_A$ that has volume $V> V_\star$ must be entangled. Note that entangled states can have $V \leq V_\star$.

Since $V > V_\star$  can only be attained by entangled states, whilst zero discord states have one-dimensional (degenerate) ellipsoids, we see that non-zero volume, or ``obesity'', is strictly stronger than discord but strictly weaker than entanglement.

\textit{Conclusion.} The quantum steering ellipsoid provides a faithful representation of any two-qubit state and a natural geometric classification of states. It yields clear and intuitive understanding into the usual key aspects of two-qubit states, uncovers surprising new features (such as the nested tetrahedron condition, skew and obesity, and incomplete steering) while prompting novel questions, such as: can we use (\ref{DetEntCond}) to define a class of ``least-classical'' separable states for fixed $(\boldsymbol{a}, \boldsymbol{b}, \boldsymbol{c})$? Can we use the nested tetrahedron condition to provide a simple construction for the best separable approximation \cite{Lewenstein} for a state $\rho$? What is the geometric characterisation of the possible steering ellipsoids for an arbitrary two-qubit state? This would potentially be useful when using ellipsoids to visualize the results of two-qubit state tomography. In Appendix~\ref{sec_extend} we have provided a discussion of several extensions of this work, beyond the two-qubit scenario to higher dimensional systems.

\begin{acknowledgments}
We wish to acknowledge Antony Milne for his useful comments, and Zuzana Gavorova and Peter Lewis for their early contributions to this topic. MP and SJ supported by the
EPSRC; SJ funded by EPSRC grant EP/K022512/1; DJ funded by the Royal Commission for the Exhibition of 1851; TR supported by the Leverhulme Trust. Research at Perimeter Institute is supported in part by the Government of Canada through NSERC and by the Province of Ontario through MRI.

\end{acknowledgments}
\bibliography{SteerEll_bib_SJ_041012}

\begin{thebibliography}{27}%
\makeatletter
\providecommand \@ifxundefined [1]{%
 \@ifx{#1\undefined}
}%
\providecommand \@ifnum [1]{%
 \ifnum #1\expandafter \@firstoftwo
 \else \expandafter \@secondoftwo
 \fi
}%
\providecommand \@ifx [1]{%
 \ifx #1\expandafter \@firstoftwo
 \else \expandafter \@secondoftwo
 \fi
}%
\providecommand \natexlab [1]{#1}%
\providecommand \enquote  [1]{``#1''}%
\providecommand \bibnamefont  [1]{#1}%
\providecommand \bibfnamefont [1]{#1}%
\providecommand \citenamefont [1]{#1}%
\providecommand \href@noop [0]{\@secondoftwo}%
\providecommand \href [0]{\begingroup \@sanitize@url \@href}%
\providecommand \@href[1]{\@@startlink{#1}\@@href}%
\providecommand \@@href[1]{\endgroup#1\@@endlink}%
\providecommand \@sanitize@url [0]{\catcode `\\12\catcode `\$12\catcode
  `\&12\catcode `\#12\catcode `\^12\catcode `\_12\catcode `\%12\relax}%
\providecommand \@@startlink[1]{}%
\providecommand \@@endlink[0]{}%
\providecommand \url  [0]{\begingroup\@sanitize@url \@url }%
\providecommand \@url [1]{\endgroup\@href {#1}{\urlprefix }}%
\providecommand \urlprefix  [0]{URL }%
\providecommand \Eprint [0]{\href }%
\providecommand \doibase [0]{http://dx.doi.org/}%
\providecommand \selectlanguage [0]{\@gobble}%
\providecommand \bibinfo  [0]{\@secondoftwo}%
\providecommand \bibfield  [0]{\@secondoftwo}%
\providecommand \translation [1]{[#1]}%
\providecommand \BibitemOpen [0]{}%
\providecommand \bibitemStop [0]{}%
\providecommand \bibitemNoStop [0]{.\EOS\space}%
\providecommand \EOS [0]{\spacefactor3000\relax}%
\providecommand \BibitemShut  [1]{\csname bibitem#1\endcsname}%
\let\auto@bib@innerbib\@empty
%</preamble>
\bibitem [{\citenamefont {Avron}\ \emph {et~al.}(2007)\citenamefont {Avron},
  \citenamefont {Bisker},\ and\ \citenamefont {Kenneth}}]{Vis2Qubits}%
  \BibitemOpen
  \bibfield  {author} {\bibinfo {author} {\bibfnamefont {J.~E.}\ \bibnamefont
  {Avron}}, \bibinfo {author} {\bibfnamefont {G.}~\bibnamefont {Bisker}}, \
  and\ \bibinfo {author} {\bibfnamefont {O.}~\bibnamefont {Kenneth}},\ }\href
  {\doibase 10.1063/1.2795217} {\bibfield  {journal} {\bibinfo  {journal} {J.\
  Math.\ Phys.}\ }\textbf {\bibinfo {volume} {48}},\ \bibinfo {pages} {102107}
  (\bibinfo {year} {2007})},\ \Eprint {http://arxiv.org/abs/arXiv:0706.2466}
  {arXiv:0706.2466} \BibitemShut {NoStop}%
\bibitem [{\citenamefont {Horodecki}\ and\ \citenamefont
  {Horodecki}(1996)}]{Tstates}%
  \BibitemOpen
  \bibfield  {author} {\bibinfo {author} {\bibfnamefont {R.}~\bibnamefont
  {Horodecki}}\ and\ \bibinfo {author} {\bibfnamefont {M.}~\bibnamefont
  {Horodecki}},\ }\href {\doibase 10.1103/PhysRevA.54.1838} {\bibfield
  {journal} {\bibinfo  {journal} {Phys. Rev. A}\ }\textbf {\bibinfo {volume}
  {54}},\ \bibinfo {pages} {1838} (\bibinfo {year} {1996})},\ \Eprint
  {http://arxiv.org/abs/quant-ph/9607007} {quant-ph/9607007} \BibitemShut
  {NoStop}%
\bibitem [{\citenamefont {Bengtsson}\ and\ \citenamefont
  {\.{Z}yczkowski}(2006)}]{Geometry}%
  \BibitemOpen
  \bibfield  {author} {\bibinfo {author} {\bibfnamefont {I.}~\bibnamefont
  {Bengtsson}}\ and\ \bibinfo {author} {\bibfnamefont {K.}~\bibnamefont
  {\.{Z}yczkowski}},\ }\href@noop {} {\emph {\bibinfo {title} {Geometry of
  Quantum States: An Introduction to Quantum Entanglement}}}\ (\bibinfo
  {publisher} {Cambridge University Press},\ \bibinfo {year}
  {2006})\BibitemShut {NoStop}%
\bibitem [{\citenamefont {Schr{\"{o}}dinger}(1935)}]{Schrodinger}%
  \BibitemOpen
  \bibfield  {author} {\bibinfo {author} {\bibfnamefont {E.}~\bibnamefont
  {Schr{\"{o}}dinger}},\ }\href {\doibase 10.1017/S0305004100013554} {\bibfield
   {journal} {\bibinfo  {journal} {Proc.\ Cam.\ Phil.\ Soc.}\ }\textbf
  {\bibinfo {volume} {31}},\ \bibinfo {pages} {555} (\bibinfo {year}
  {1935})}\BibitemShut {NoStop}%
\bibitem [{\citenamefont {Gisin}(1996)}]{Gisin}%
  \BibitemOpen
  \bibfield  {author} {\bibinfo {author} {\bibfnamefont {N.}~\bibnamefont
  {Gisin}},\ }\href {\doibase 10.1016/S0375-9601(96)80001-6} {\bibfield
  {journal} {\bibinfo  {journal} {Phys. Lett. A}\ }\textbf {\bibinfo {volume}
  {210}},\ \bibinfo {pages} {151} (\bibinfo {year} {1996})}\BibitemShut
  {NoStop}%
\bibitem [{\citenamefont {Hughston}\ \emph {et~al.}(1993)\citenamefont
  {Hughston}, \citenamefont {Jozsa},\ and\ \citenamefont {Wootters}}]{HJW}%
  \BibitemOpen
  \bibfield  {author} {\bibinfo {author} {\bibfnamefont {L.~P.}\ \bibnamefont
  {Hughston}}, \bibinfo {author} {\bibfnamefont {R.}~\bibnamefont {Jozsa}}, \
  and\ \bibinfo {author} {\bibfnamefont {W.~K.}\ \bibnamefont {Wootters}},\
  }\href {\doibase 10.1016/0375-9601(93)90880-9} {\bibfield  {journal}
  {\bibinfo  {journal} {Phys. Lett. A}\ }\textbf {\bibinfo {volume} {183}},\
  \bibinfo {pages} {14} (\bibinfo {year} {1993})}\BibitemShut {NoStop}%
\bibitem [{\citenamefont {Spekkens}\ and\ \citenamefont
  {Rudolph}(2002)}]{Rob-Terry}%
  \BibitemOpen
  \bibfield  {author} {\bibinfo {author} {\bibfnamefont {R.~W.}\ \bibnamefont
  {Spekkens}}\ and\ \bibinfo {author} {\bibfnamefont {T.}~\bibnamefont
  {Rudolph}},\ }\href {http://www.rintonpress.com/journals/qiconline.html#v2n1}
  {\bibfield  {journal} {\bibinfo  {journal} {Quantum Inform. Compu.}\ }\textbf
  {\bibinfo {volume} {2}},\ \bibinfo {pages} {66} (\bibinfo {year} {2002})},\
  \Eprint {http://arxiv.org/abs/quant-ph/0107042} {quant-ph/0107042}
  \BibitemShut {NoStop}%
\bibitem [{\citenamefont {Verstraete}(2002)}]{FrankPhD}%
  \BibitemOpen
  \bibfield  {author} {\bibinfo {author} {\bibfnamefont {F.}~\bibnamefont
  {Verstraete}},\ }\href@noop {} {Ph.D. thesis},\ \bibinfo  {school}
  {Katholieke Universiteit Leuven} (\bibinfo {year} {2002})\BibitemShut
  {NoStop}%
\bibitem [{\citenamefont {Verstraete}\ \emph {et~al.}(2001)\citenamefont
  {Verstraete}, \citenamefont {Dehaene},\ and\ \citenamefont
  {DeMoor}}]{FrankFilter}%
  \BibitemOpen
  \bibfield  {author} {\bibinfo {author} {\bibfnamefont {F.}~\bibnamefont
  {Verstraete}}, \bibinfo {author} {\bibfnamefont {J.}~\bibnamefont {Dehaene}},
  \ and\ \bibinfo {author} {\bibfnamefont {B.}~\bibnamefont {DeMoor}},\ }\href
  {\doibase 10.1103/PhysRevA.64.010101} {\bibfield  {journal} {\bibinfo
  {journal} {Phys.\ Rev.\ A}\ }\textbf {\bibinfo {volume} {64}},\ \bibinfo
  {pages} {010101} (\bibinfo {year} {2001})},\ \Eprint
  {http://arxiv.org/abs/quant-ph/0011111} {quant-ph/0011111} \BibitemShut
  {NoStop}%
\bibitem [{\citenamefont {Altepeter}\ \emph {et~al.}(2009)\citenamefont
  {Altepeter}, \citenamefont {Jeffrey}, \citenamefont {Medic},\ and\
  \citenamefont {Kumar}}]{Altepeter}%
  \BibitemOpen
  \bibfield  {author} {\bibinfo {author} {\bibfnamefont {J.~B.}\ \bibnamefont
  {Altepeter}}, \bibinfo {author} {\bibfnamefont {E.~R.}\ \bibnamefont
  {Jeffrey}}, \bibinfo {author} {\bibfnamefont {M.}~\bibnamefont {Medic}}, \
  and\ \bibinfo {author} {\bibfnamefont {P.}~\bibnamefont {Kumar}},\ }in\ \href
  {\doibase 10.1364/IQEC.2009.IWC1} {\emph {\bibinfo {booktitle} {Conference on
  Lasers and Electro-Optics/International Quantum Electronics Conference}}}\
  (\bibinfo  {publisher} {Optical Society of America},\ \bibinfo {year}
  {2009})\ p.\ \bibinfo {pages} {IWC1}\BibitemShut {NoStop}%
\bibitem [{\citenamefont {Peres}(1996)}]{PeresPT}%
  \BibitemOpen
  \bibfield  {author} {\bibinfo {author} {\bibfnamefont {A.}~\bibnamefont
  {Peres}},\ }\href {\doibase 10.1103/PhysRevLett.77.1413} {\bibfield
  {journal} {\bibinfo  {journal} {Phys.\ Rev.\ Lett.}\ }\textbf {\bibinfo
  {volume} {77}},\ \bibinfo {pages} {1413} (\bibinfo {year} {1996})},\ \Eprint
  {http://arxiv.org/abs/quant-ph/9604005} {quant-ph/9604005} \BibitemShut
  {NoStop}%
\bibitem [{\citenamefont {Horodecki}\ \emph {et~al.}(1996)\citenamefont
  {Horodecki}, \citenamefont {Horodecki},\ and\ \citenamefont
  {Horodecki}}]{HorodeckiPT}%
  \BibitemOpen
  \bibfield  {author} {\bibinfo {author} {\bibfnamefont {M.}~\bibnamefont
  {Horodecki}}, \bibinfo {author} {\bibfnamefont {P.}~\bibnamefont
  {Horodecki}}, \ and\ \bibinfo {author} {\bibfnamefont {R.}~\bibnamefont
  {Horodecki}},\ }\href {\doibase 10.1016/S0375-9601(96)00706-2} {\bibfield
  {journal} {\bibinfo  {journal} {Phys.\ Lett.\ A}\ }\textbf {\bibinfo {volume}
  {223}},\ \bibinfo {pages} {1} (\bibinfo {year} {1996})},\ \Eprint
  {http://arxiv.org/abs/quant-ph/9605038} {quant-ph/9605038} \BibitemShut
  {NoStop}%
\bibitem [{\citenamefont {Sanpera}\ \emph {et~al.}(1998)\citenamefont
  {Sanpera}, \citenamefont {Tarrach},\ and\ \citenamefont {Vidal}}]{Sanpera}%
  \BibitemOpen
  \bibfield  {author} {\bibinfo {author} {\bibfnamefont {A.}~\bibnamefont
  {Sanpera}}, \bibinfo {author} {\bibfnamefont {R.}~\bibnamefont {Tarrach}}, \
  and\ \bibinfo {author} {\bibfnamefont {G.}~\bibnamefont {Vidal}},\ }\href
  {\doibase 10.1103/PhysRevA.58.826} {\bibfield  {journal} {\bibinfo  {journal}
  {Phys.\ Rev.\ A}\ }\textbf {\bibinfo {volume} {58}},\ \bibinfo {pages} {826}
  (\bibinfo {year} {1998})}\BibitemShut {NoStop}%
\bibitem [{\citenamefont {Wootters}(1998)}]{PhysRevLett.80.2245}%
  \BibitemOpen
  \bibfield  {author} {\bibinfo {author} {\bibfnamefont {W.~K.}\ \bibnamefont
  {Wootters}},\ }\href {\doibase 10.1103/PhysRevLett.80.2245} {\bibfield
  {journal} {\bibinfo  {journal} {Phys. Rev. Lett.}\ }\textbf {\bibinfo
  {volume} {80}},\ \bibinfo {pages} {2245} (\bibinfo {year}
  {1998})}\BibitemShut {NoStop}%
\bibitem [{\citenamefont {Vidrighin}(2012)}]{Mihai}%
  \BibitemOpen
  \bibfield  {author} {\bibinfo {author} {\bibfnamefont {M.}~\bibnamefont
  {Vidrighin}},\ }\href {http://mathoverflow.net/questions/109153} {\enquote
  {\bibinfo {title} {mathoverflow.net/questions/109153},}\ } (\bibinfo {year}
  {2012})\BibitemShut {NoStop}%
\bibitem [{\citenamefont {Ollivier}\ and\ \citenamefont
  {Zurek}(2001)}]{Discord}%
  \BibitemOpen
  \bibfield  {author} {\bibinfo {author} {\bibfnamefont {H.}~\bibnamefont
  {Ollivier}}\ and\ \bibinfo {author} {\bibfnamefont {W.~H.}\ \bibnamefont
  {Zurek}},\ }\href {\doibase 10.1103/PhysRevLett.88.017901} {\bibfield
  {journal} {\bibinfo  {journal} {Phys.\ Rev.\ Lett.}\ }\textbf {\bibinfo
  {volume} {88}},\ \bibinfo {pages} {017901} (\bibinfo {year}
  {2001})}\BibitemShut {NoStop}%
\bibitem [{Note1()}]{Note1}%
  \BibitemOpen
  \bibinfo {note} {Discord can be calculated analytically when $\theta
  =0,\protect \frac {\pi }{2}$, and at these points $\rho $ is an X-state \cite
  {ShiXstateDiscord} and this agrees with our findings.}\BibitemShut {Stop}%
\bibitem [{\citenamefont {Mingjun~Shi}\ and\ \citenamefont
  {Du}(2011)}]{ShiXstateDiscord}%
  \BibitemOpen
  \bibfield  {author} {\bibinfo {author} {\bibfnamefont {C.~S.}\ \bibnamefont
  {Mingjun~Shi}, \bibfnamefont {Fengjian~Jiang}}\ and\ \bibinfo {author}
  {\bibfnamefont {J.}~\bibnamefont {Du}},\ }\href {\doibase
  10.1088/1367-2630/13/7/073016} {\bibfield  {journal} {\bibinfo  {journal}
  {New J. Phys.}\ }\textbf {\bibinfo {volume} {13}},\ \bibinfo {pages} {073016}
  (\bibinfo {year} {2011})},\ \Eprint {http://arxiv.org/abs/arXiv:1102.4888}
  {arXiv:1102.4888} \BibitemShut {NoStop}%
\bibitem [{Note2()}]{Note2}%
  \BibitemOpen
  \bibinfo {note} {More formally, it was shown \cite {FJohn, Ball} that every
  convex body in $\protect \mathbb {R}^n$ Euclidean space contains a unique
  maximal volume ellipsoid, and that a Euclidean ball is that maximum if and
  only if there exist vectors $\protect \{\protect \bm {x}_k\protect \}$ in the
  set of contact points of the ball with the convex body, and positive numbers
  $\protect \{c_k\protect \}$ such that $\DOTSB \sum@ \slimits@ _k c_k \protect
  \bm {x}_k = 0$ and $\DOTSB \sum@ \slimits@ _k c_k \protect \bm {x}^T_k
  \protect \bm {x}_k =r^2 \protect \mathbbm {1} $, where $r$ is the radius of
  the ball. Inspection of the in-sphere for the regular tetrahedron in the
  Bloch sphere shows that is is the case for $\protect \{\protect \bm
  {x}_k\protect \}$ being the four intersection points.}\BibitemShut {Stop}%
\bibitem [{\citenamefont {John}(1948)}]{FJohn}%
  \BibitemOpen
  \bibfield  {author} {\bibinfo {author} {\bibfnamefont {F.}~\bibnamefont
  {John}},\ }in\ \href@noop {} {\emph {\bibinfo {booktitle} {Studies and Essays
  Presented to R. Courant on his 60th Birthday, Jan. 8, 1948}}}\ (\bibinfo
  {publisher} {Interscience},\ \bibinfo {year} {1948})\ pp.\ \bibinfo {pages}
  {187--204}\BibitemShut {NoStop}%
\bibitem [{\citenamefont {Ball}(1992)}]{Ball}%
  \BibitemOpen
  \bibfield  {author} {\bibinfo {author} {\bibfnamefont {K.}~\bibnamefont
  {Ball}},\ }\href {\doibase 10.1007/BF00182424} {\bibfield  {journal}
  {\bibinfo  {journal} {Geometriae Dedicata}\ }\textbf {\bibinfo {volume}
  {41}},\ \bibinfo {pages} {241} (\bibinfo {year} {1992})}\BibitemShut
  {NoStop}%
\bibitem [{\citenamefont {Lewenstein}\ and\ \citenamefont
  {Sanpera}(1998)}]{Lewenstein}%
  \BibitemOpen
  \bibfield  {author} {\bibinfo {author} {\bibfnamefont {M.}~\bibnamefont
  {Lewenstein}}\ and\ \bibinfo {author} {\bibfnamefont {A.}~\bibnamefont
  {Sanpera}},\ }\href {\doibase 10.1103/PhysRevLett.80.2261} {\bibfield
  {journal} {\bibinfo  {journal} {Phys. Rev. Lett.}\ }\textbf {\bibinfo
  {volume} {80}},\ \bibinfo {pages} {2261} (\bibinfo {year} {1998})},\ \Eprint
  {http://arxiv.org/abs/quant-ph/9707043} {quant-ph/9707043} \BibitemShut
  {NoStop}%
\bibitem [{\citenamefont {Frosch}(2006)}]{lorentz}%
  \BibitemOpen
  \bibfield  {author} {\bibinfo {author} {\bibfnamefont {R.}~\bibnamefont
  {Frosch}},\ }\href@noop {} {\emph {\bibinfo {title} {Four-tensors, the Mother
  Tongue of Classical Physics}}}\ (\bibinfo  {publisher} {vdf Hochschulverlag
  AG},\ \bibinfo {year} {2006})\BibitemShut {NoStop}%
\bibitem [{\citenamefont {Daki{\'{c}}}\ \emph {et~al.}(2010)\citenamefont
  {Daki{\'{c}}}, \citenamefont {Vedral},\ and\ \citenamefont
  {Brukner}}]{VlatkoDiscord}%
  \BibitemOpen
  \bibfield  {author} {\bibinfo {author} {\bibfnamefont {B.}~\bibnamefont
  {Daki{\'{c}}}}, \bibinfo {author} {\bibfnamefont {V.}~\bibnamefont {Vedral}},
  \ and\ \bibinfo {author} {\bibfnamefont {{\v{C}}.}~\bibnamefont {Brukner}},\
  }\href {\doibase 10.1103/PhysRevLett.105.190502} {\bibfield  {journal}
  {\bibinfo  {journal} {Phys.\ Rev.\ Lett.}\ }\textbf {\bibinfo {volume}
  {105}},\ \bibinfo {pages} {190502} (\bibinfo {year} {2010})},\ \Eprint
  {http://arxiv.org/abs/arXiv:1004.0190} {arXiv:1004.0190} \BibitemShut
  {NoStop}%
\bibitem [{\citenamefont {{Goyal}}\ \emph {et~al.}(2011)\citenamefont
  {{Goyal}}, \citenamefont {{Neethi Simon}}, \citenamefont {{Singh}},\ and\
  \citenamefont {{Simon}}}]{qutrit}%
  \BibitemOpen
  \bibfield  {author} {\bibinfo {author} {\bibfnamefont {S.~K.}\ \bibnamefont
  {{Goyal}}}, \bibinfo {author} {\bibfnamefont {B.}~\bibnamefont {{Neethi
  Simon}}}, \bibinfo {author} {\bibfnamefont {R.}~\bibnamefont {{Singh}}}, \
  and\ \bibinfo {author} {\bibfnamefont {S.}~\bibnamefont {{Simon}}},\
  }\href@noop {} {\enquote {\bibinfo {title} {{Geometry of the generalized
  {B}loch sphere for qutrit}},}\ } (\bibinfo {year} {2011}),\ \Eprint
  {http://arxiv.org/abs/arXiv:1111.4427} {arXiv:1111.4427} \BibitemShut
  {NoStop}%
\bibitem [{\citenamefont {Chen}\ and\ \citenamefont
  {\DJ{}okovi\'{c}}(2012)}]{qubitqudit}%
  \BibitemOpen
  \bibfield  {author} {\bibinfo {author} {\bibfnamefont {L.}~\bibnamefont
  {Chen}}\ and\ \bibinfo {author} {\bibfnamefont {D.~{\v{Z}}.}\ \bibnamefont
  {\DJ{}okovi\'{c}}},\ }\href {\doibase 10.1103/PhysRevA.86.062332} {\bibfield
  {journal} {\bibinfo  {journal} {Phys. Rev. A}\ }\textbf {\bibinfo {volume}
  {86}},\ \bibinfo {pages} {062332} (\bibinfo {year} {2012})},\ \Eprint
  {http://arxiv.org/abs/arXiv:1210.0111} {arXiv:1210.0111} \BibitemShut
  {NoStop}%
\bibitem [{\citenamefont {Milne}\ \emph {et~al.}(2014)\citenamefont {Milne},
  \citenamefont {Jevtic}, \citenamefont {Jennings}, \citenamefont {Wiseman},\
  and\ \citenamefont {Rudolph}}]{Antony}%
  \BibitemOpen
  \bibfield  {author} {\bibinfo {author} {\bibfnamefont {A.}~\bibnamefont
  {Milne}}, \bibinfo {author} {\bibfnamefont {S.}~\bibnamefont {Jevtic}},
  \bibinfo {author} {\bibfnamefont {D.}~\bibnamefont {Jennings}}, \bibinfo
  {author} {\bibfnamefont {H.}~\bibnamefont {Wiseman}}, \ and\ \bibinfo
  {author} {\bibfnamefont {T.}~\bibnamefont {Rudolph}},\ }\href@noop {}
  {\enquote {\bibinfo {title} {Quantum steering ellipsoids, extremal physical
  states and monogamy},}\ } (\bibinfo {year} {2014}),\ \Eprint
  {http://arxiv.org/abs/arXiv:1403.0418} {arXiv:1403.0418} \BibitemShut
  {NoStop}%
\end{thebibliography}%
\appendix
\section{\label{sec_SteerEll_Deriv}The steering ellipsoid for a general two qubit state}

Previous approaches to representing two qubit states have included partitioning the set of all two qubit states into SLOCC (Stochastic Local Operations and Classical Communication) equivalence classes \cite{Vis2Qubits}, which results in a three dimensional representation of a state $\rho$ through its SLOCC orbit, defined as
\begin{equation}
\mathcal{S}(\rho):=\left\{ \frac{  S_A \otimes S_B \rho( S_A \otimes S_B)^{\dag}  }{  \tr( S_A \otimes S_B \rho ( S_A \otimes S_B)^{\dag})  } : \,  S_A,S_B \in \mathrm{GL}(2,\mathbb{C}) \right \}\nonumber
\end{equation}
with $\mathrm{GL}(2,\mathbb{C})$ the group of invertible, complex 2 $\times$ 2 matrices. However replacing $\rho$ with its SLOCC orbit is far from faithful, and amounts to a highly coarse-grained representation of the state that erases much of its detail. Another approach is to start with a Pauli basis expansion of $\rho$, which can be converted via a state-dependent choice of local unitaries on both qubits, to a representation involving three spatial vectors \cite{HorodeckiPT}. Again, this is still not a faithful representation, and more importantly it is extremely difficult to develop any intuition for what the vector representing correlations actually means.

In this section we provide the details for constructing the steering ellipsoid representation of an arbitrary two qubit quantum state.

Consider a two-qubit state
\begin{equation}
\Theta =
\begin{pmatrix}
1 & \ve b^T\\
\ve a & T
\end{pmatrix}.
\end{equation}

The matrix $\Theta$ transforms, up to a normalization, under SLOCC operations $\rho' = S_A\otimes S_B \rho (S_A \otimes S_B)^\dag$ as $\Theta'=\Lambda_A \Theta \Lambda^T_B$ \cite{FrankFilter}
where $\Lambda_A, \Lambda_B$ are proper orthochronous Lorentz transformations given by
$\Lambda_{W} = \Upsilon S_W \otimes S_W^* \Upsilon^\dag/|\det S_W|,$ $W \in \{A,B\}$ and
\begin{equation}
\label{Seagull}
\Upsilon=\frac{1}{\sqrt{2}}
\begin{pmatrix}1 & 0 & 0 & 1 \\
    0 & 1 & 1 & 0 \\
    0 & i & -i & 0 \\
    1 & 0 & 0 & -1
    \end{pmatrix}.
\end{equation}
In particular, local unitaries rotate the Pauli basis: $\boldsymbol{a}\rightarrow O_A\boldsymbol{a}$, $\boldsymbol{b}\rightarrow O_B\boldsymbol{b}$, $T\rightarrow O_AT O_B^T$, where $O_A,O_B \in \mathrm{SO}(3)$.

If $b = 1$ then $\Theta$ must be a product state and no steering can occur, so assume otherwise. Define $\gamma = 1/\sqrt{1-b^2}$, and the ``canonical'' filtered state
\begin{equation}
\can\Theta = \gamma \Theta L_{\ve b}, \label{thetapdef}
\end{equation}
where $L_{\ve b}$ is a Lorentz boost by $\ve b$ \cite{lorentz}:
\begin{equation}
L_{\ve b} = \begin{pmatrix}
\gamma & -\gamma \ve b^T \\
-\gamma \ve b & \I + \frac{\gamma - 1}{b^2}\ve b \ve b^T
\end{pmatrix}.
\end{equation}

The $\gamma$ in \eqref{thetapdef} ensures that $\can\Theta$ is normalized: the top-left element is
\begin{equation}
\gamma \begin{pmatrix} 1 & \ve b^T\end{pmatrix} \begin{pmatrix}\gamma \\-\gamma\ve b\end{pmatrix} = \gamma^2(1-b^2) = 1.
\end{equation}

This particular boost is special as it leaves Bob's reduced state is maximally mixed: the top-right block of $\can\Theta$ is
\begin{multline}
\gamma \begin{pmatrix} 1 & \ve b^T\end{pmatrix} \begin{pmatrix}-\gamma \ve b^T \\ \I + \frac{\gamma - 1}{b^2}\ve b \ve b^T
\end{pmatrix} = \\\gamma(-\gamma \ve b^T + \ve b^T + (\gamma - 1)\ve b^T) = \ve 0^T.
\end{multline}

By the above two observations we can write
\begin{equation}
  \can\Theta= \begin{pmatrix}1 & \ve 0^T \\ \can{\ve a} & \can T \end{pmatrix}.\label{thetapform}
\end{equation}

The set of states Bob can steer Alice to will be exactly the same for $\Theta$ and any $\Theta' = \Theta \Lambda_B$ (this includes the special case $\Lambda_B = \gamma L_{\ve b}$): $Y = \frac12 \Theta X \iff Y = \frac12 \Theta' X'$ where $X' = \Lambda^{-1}_B X$. When $X, X'$ are viewed as 4-vectors in Minkowski space, then $X$ corresponds to a positive operator iff $X'$ does because $\Lambda^{-1}_B$ preserves the forward light cone.

In the instance when Alice and Bob share the canonical state, $\Theta' = \can\Theta$, what is the set of states that Alice's qubit can be steered to? Writing $X = \begin{pmatrix}t \\ \ve x\end{pmatrix}$ we see that without loss of generality we can take $t=1$ since the effect of multiplying $X$ by a positive number is undone when normalizing $Y$. Hence the positivity condition becomes $x \leq 1$. If $x<1$ then we could write $X$ as a convex combination of ones with $x=1$, so we restrict our attention to the latter case. So any state $\ve y$ that Bob can steer Alice to is a convex combination of states of the form
$\ve y = \can{\ve a} + \can{T}\ve x$
with $x = 1$. But this defines a linear image of the unit sphere of all $\ve x$ displaced by $\can{\ve a}$, or in other words an ellipsoid with centre $\ve{c}_A = \can{\ve a}$.

Consider the singular value decomposition $\can T = O_1 D O_2$. $O_2$ simply rotates and/or reflects the unit sphere of all $\ve{x}$ and so can be ignored when thinking about the ellipsoid properties. $D$ rescales the sphere and thus gives the lengths of the semi-axes of the resulting ellipsoid. $O_1$ rotates the semi-axes. We have $\can T \can T^{T} = O_1 D^2 O_1^T$ and so the lengths of the semi-axes can also be found by square rooting the eigenvalues of $\can T \can T^T$ whilst their directions can be found from its eigenvectors.

Combining \eqref{thetapdef} with \eqref{thetapform} we find
\begin{equation}
\can {\ve a} = \gamma\begin{pmatrix}\ve a & T\end{pmatrix}\begin{pmatrix}\gamma \\ -\gamma \ve b\end{pmatrix} = \gamma^2(\ve a - T\ve b):=\ve{c}_A,\label{Ell_centre}
\end{equation}
\begin{eqnarray}
\can T &=& \gamma\begin{pmatrix}\ve a & T\end{pmatrix}\begin{pmatrix}-\gamma \ve b^T \I + \frac{\gamma - 1}{b^2}\ve b \ve b^T
\end{pmatrix}\\ &=& \gamma\left(-\gamma \ve a \ve b^T + T + \frac{\gamma - 1}{b^2}T\ve b \ve b^T\right).
\end{eqnarray}

And so, after some algebra,
\begin{align}
\can T \can T^T = \gamma^2(T T^T - \ve a \ve a^T) + \can{\ve a} \can{\ve a}^T.
\end{align}

%\begin{multline}
%\begin{widetext}
%\begin{eqnarray}
%T' T'^T &=& \gamma^2\left(
%\gamma^2 b^2 \ve a \ve a^T - \gamma \ve a \ve b^T T^T - \gamma(\gamma - 1)\ve a \ve b^T T^T
%- \gamma T\ve b \ve a^T + TT^T \right) + \gamma^2\left( \frac{\gamma - 1}{b^2}T \ve b \ve b^T T^T
%+ (\gamma - 1)\left(-\gamma T\ve b \ve a^T + \frac{T\ve b \ve b^T T^T + (\gamma-1)T\ve b \ve b^T T^T}{b^2}\right)
%\right)
%\\ &=& \gamma^2\left(\gamma^2 (b^2 \ve a \ve a^T - \ve a \ve b^T T^T -T\ve b \ve a^T + T\ve b \ve b^T T^T) + T T^T \right)
%\\ &=& \gamma^2(T T^T - \ve a \ve a^T) + \ve a' \ve a'^T.
%\end{eqnarray}
%\end{widetext}
%\end{multline}

Using $-\ve a \ve b^T \left(\I + \frac{\gamma-1}{b^2}\ve b \ve b^T\right) = -\gamma \ve a \ve b^T$ we can also write
\begin{equation}
\can T = \gamma\left(T - \ve a \ve b^T\right)\left(\I + \frac{\gamma-1}{b^2}\ve b \ve b^T\right)\label{alternativeT}
\end{equation}
leading to the form in the main text:
\begin{equation}
\can T \can T^T = \gamma^2(T-\ve a \ve b^T)(\I + \gamma^2 \ve b \ve b^T)(T^T - \ve b \ve a^T) =: Q_A
\label{eq_QA}
\end{equation}

%At no point have we used any inverses or made any assumptions about ranks and so the above works in all cases.
Since $\gamma L_{\ve b}$ is invertible we have $\rank(\can\Theta) = \rank(\Theta)$. By counting linearly independent columns in \eqref{thetapform} we have $\rank(\can \Theta) = \rank(\can T) + 1$ thus proving that the dimension of the ellipsoid is $\rank(\Theta) - 1$.

%The volume of Alice's ellipsoid is $\frac{4\pi}3$ times the product of her semiaxes, i.e. $\abs{\det(T')}$. By expanding $\det(\Theta')$ along the first row of \eqref{thetapform} we see that $\det(\Theta') = \det(T')$. Since $\det(L_{\ve b}) = 1$ we have $\det(\Theta') = \gamma^4\det(\Theta)$. Therefore:
%\begin{equation}
%V_A = \frac{4\pi}3 \gamma^4 \abs{\det(\Theta)}
%\end{equation}

%Using the block matrix formula\footnote{\url{http://en.wikipedia.org/wiki/Determinant\#Block\_matrices}}
%\begin{equation}
%\det\begin{pmatrix}A & B \\ C & D\end{pmatrix} = \det(A)\det(D - CA^{-1}B)
%\end{equation}
%we obtain $\det(\Theta) = \det(T - \ve a \ve b^T)$ giving the alternative form of $V_A$. The same expression can be obtained by claculating $\det(T')$ directly from \eqref{alternativeT} and Sylvester's determinant theorem.

%\newpage

\section{\label{sec_RecRho}Reconstructing the state from its steering ellipsoid and the Bloch vectors}
We will need the following simple property of orthogonal matrices.
\begin{lem}
Suppose $O_1, O_2 \in O(n)$, $P$ is a projector and $\ve b$ a vector, with
\begin{equation}
PO_1\ve b = PO_2\ve b.\label{lemcond}
\end{equation}
Then there exists a $O_3 \in O(n)$ with
\begin{equation}
O_3 \ve b = \ve b \label{lemconc1}
\end{equation}
and
\begin{equation}
PO_1O_3 = PO_2. \label{lemconc2}
\end{equation}
\label{orthlem}
\end{lem}
\begin{proof}
Denote $Q = \I - P$. Notice that $\nrm{O_1\ve b} = \nrm{\ve b} = \nrm{O_2\ve b}$ and $\nrm{PO_1\ve b} = \nrm{PO_2\ve b}$ by equation~\eqref{lemcond}. Since $\nrm{\ve v}^2 = \nrm{P\ve v}^2 + \nrm{Q \ve v}^2$ for any $\ve v$, we have $\nrm{QO_1\ve b} = \nrm{QO_2 \ve b}$. Hence there exists an $O_4 \in O(n)$ that acts trivially in the support of $P$, i.e. $PO_4 = O_4P = P$, and sends $QO_2\ve b$ to $QO_1\ve b$. Define
\begin{equation}
O_3 = O_1^T O_4 O_2.
\end{equation}
Then
\begin{equation}
PO_1O_3 = PO_1 O_1^T O_4 O_2 = PO_4 O_2 = PO_2
\end{equation}
as required, and
\begin{multline}
O_3\ve b = O_1^T O_4 O_2\ve b = O_1^T (P+Q) O_4 (P+Q) O_2 \ve b \\= O_1^T(
P O_4 P +
P O_4 Q +
Q O_4 P +
Q O_4 Q)O_2\ve b =\\
O_1^T(P^2 +
PQ +
QP +
Q O_4 Q)O_2\ve b  = O_1^T (P + Q O_4 Q)O_2 \ve b
\end{multline}
since $P$ and $Q$ are orthogonal projectors. Using equation~\eqref{lemcond} and $O_4 Q O_2\ve b = Q O_1 \ve b$ we find
\begin{equation}
O_3 \ve b = O_1^T (P + Q^2)O_1 \ve b = O_1^T(P+Q)O_1 \ve b = O_1^T O_1 \ve b = \ve b
\end{equation}
as required.
\end{proof}

Suppose we know $\elip_A$ (i.e. $Q_A$ and $\ve c_A$) and the Bloch vectors $\ve a, \ve b$ for some two-qubit state $\Theta$. Defining $\can\Theta$ as the canonical state, equation \eqref{thetapform} in section \ref{sec_SteerEll_Deriv}, we know that $\can{\ve a} = \ve c_A$ and $\can T \can T^T = Q_A$ from equations \eqref{Ell_centre} and \eqref{eq_QA} respectively. Hence $\can T = \sqrt{Q_A}O_1$ for some $O_1 \in O(3)$. By inverting the Lorentz boost $\Theta = \frac{1}{\gamma}\can\Theta L_{\ve{-b}}$ (c.f. equation \eqref{thetapdef}) we see that $O_1$ must satisfy $\ve a = \ve c_A + \sqrt{Q_A}O_1\ve b$. Working out the inverted form of $T$ leads to the equation \eqref{Teq} given in the main text.

We now show that if we choose $\can T' = \sqrt{Q_A}O_2$ for an \emph{arbitrary} $O_2 \in O(3)$ satisfying $\ve a = \ve c_A + \sqrt{Q_A}O_2\ve b$ then the ``reconstructed'' two-qubit state $\Theta'$ differs at most from $\Theta$ by a local unitary and/or partial transpose on Bob's qubit. (Since entangled two-qubit states have negative partial transpose, they do not in fact have the partial transpose ambiguity.)

Since $\sqrt{Q_A} O_1 \ve b = \sqrt{Q_A} O_2\ve b$ we have $PO\ve b = PO'\ve b$ where $P$ projects onto the support of $Q_A$. Hence we can apply Lemma~\ref{orthlem} to obtain a third rotation $O_3 \in O(n)$. Let $s(\cdot) = \begin{pmatrix} 1 & \ve 0^T \\ \ve 0 & \cdot\end{pmatrix}$ denote its argument acting on the spatial part of space-time. Spatial rotations about a direction commute with boosts in that direction, and so equation~\eqref{lemconc1} gives $L_{-\ve b} s(O_3) = s(O_3) L_{-\ve b}$. Furthermore, equation~\eqref{lemconc2} gives $\can T O_3 = \sqrt{Q_A}O_1O_3 = \sqrt{Q_A}O_2 = \can T'$. Hence
\begin{multline}
\Theta s(O_3) = \can \Theta \frac{L_{-\ve b}}{\gamma} s(O_3) = \begin{pmatrix}1 & \ve 0^T \\ \can{\ve a} & \can TO_3\end{pmatrix}\frac{L_{-\ve b}}{\gamma} \\= \begin{pmatrix}1 & \ve 0^T \\ \can{\ve a} & \can T'\end{pmatrix}\frac{L_{-\ve b}}{\gamma} = \can \Theta'\frac{L_{-\ve b}}{\gamma} = \Theta'.
\end{multline}
Since $O_3 = O_B\Delta$ where $O_B \in SO(3)$ and $\Delta$ is either the identity matrix or $\diag(1,-1,1)$ (the action of the transpose on Bob in the Pauli basis), we have the desired local unitary.

To remove all ambiguity we need to display the additional information in $O_3$ on the ellipsoid, and this is possible if we depict (ideally in the centre of the ellipsoid) the image of the three principal axes under $O_3$. Alternatively we could use a colouring scheme such as that proposed in \cite{Altepeter}. Together with the steering ellipsoid and Bloch vectors this gives a faithful representation any two qubit state using three dimensional objects. But since the ambiguity in $O_3$ is merely a choice of local basis for Bob, all correlation properties can be found solely from the steering ellipsoid and Bloch vectors, and some properties (such as separability) are fixed by the ellipsoid alone.

\section{Proof of the complete steering theorem}\label{sec_compthm}

We present an extended formulation of the theorem:

%\begin{thm}
\textit{Theorem} Consider some non-product two-qubit state $\Theta = \begin{pmatrix}1 & \ve b^T\\\ve a & T\end{pmatrix}$ with ellipsoids $\elip_A$ and $\elip_B$. The following are equivalent:
\begin{enumerate}
\item Complete steering: for all convex decompositions of $\ve a$ into states in $\elip_A$, there exists a POVM for Bob that steers Alice to it.\label{fullsteer}
\item Any surface steering: there exists a convex decomposition of $\ve a$ into states on the surface of $\elip_A$ with a POVM for Bob that steers Alice to it.\label{surfsteer}
\item Alice's Bloch vector lies on the surface of her ellipsoid scaled down by $b$.\label{blochcond}% (In other words $\ve a$'s distance from the centre of $\elip_A$ is maximal given $\elip_A$ and $b$, or $b$ is minimal given $\ve a$ and $\elip_A$.)\label{blochcond}
\item The affine span of $\elip_B$ contains the maximally mixed state.\label{spancond}
\item $\begin{pmatrix}1 \\ \ve 0\end{pmatrix} \in \mathrm{range}(\Theta^T)$.\label{rangecond}
\item $\begin{pmatrix}1 \\ \ve 0\end{pmatrix} \in \ker(\Theta)^\perp$.\label{kercond}
\end{enumerate}
%\end{thm}

%Notice that \ref{blochcond}--\ref{kercond} are all trivially satisfied if $\elip_A$ is non-degenerate or $\ve b = \ve 0$. In the case of steering needles, \ref{blochcond} is my criterion for zero discord and \ref{spancond} is Sania's. (Aside: Just a discord classifies needle states into ``classical-classical'', ``classical-quantum'', ``quantum-classical'' and ``quantum-quantum'', where the word ``classical'' comes from the fact that a needle whose span includes the maximally mixed state can be rotated to lie on the $Z$ axis, we could classify pancakes into ``real-real'', ``real-complex'', ``compex-real'' and ``complex-complex'' since a pancake whose span includes the maximally mixed state can be rotated to lie in the $X$-$Z$ plane.)

\begin{proof} \underline{Preliminaries}  From section \ref{sec_SteerEll_Deriv}: $\Theta = \can \Theta \frac{L_{-\ve b}}{\gamma}$ and so \ref{kercond} is equivalent to $\begin{pmatrix}1 \\ \ve b\end{pmatrix} \in \ker(\can \Theta)^\perp$. Inspecting equation~\eqref{thetapform} we see that any vector in $\ker(\can \Theta)$ is of the form $\begin{pmatrix}0 \\ \ve r\end{pmatrix}$ with $\ve r \in \ker(\can T)$ and so \ref{kercond} is equivalent to $\ve b \in \ker(\can T)^\perp$.

$\elip_A$ are the points that can be written $\can {\ve a} + \can{T}\ve x$ where $x \le 1$. Therefore the surface of $\elip_A$ are the points that can be written $\can {\ve a} + \can T \ve x$ where $x = 1$ and $\ve x \in \ker(\can T)^\perp$. Hence the scaled down surface is $\can {\ve a} + \can T\ve x$ where $x = b$ and $\ve x \in \ker(\can T)^\perp$.

From $\Theta =\can \Theta \frac{L_{-\ve b}}{\gamma}$ we calculate that $\ve a = \can{\ve a} + \can T \ve b$.

\underline{\ref{fullsteer} $\implies$ \ref{surfsteer}:} Trivial.

\underline{\ref{surfsteer} $\implies$ \ref{kercond}:} Let $\ve y_i$ on the surface of $\elip_A$ form a convex decomposition  $\sum_i p_i \ve y_i = \ve a$. Since they are on the surface, we have $\ve y_i = \can{\ve a} + \can T \ve x_i$ where $x_i=1$ and $\ve x_i \in \ker(\can T)^\perp$. Suppose we also have $\ve y_i = \can{\ve a} + \can T \ve x_i'$ with $x_i' \leq 1$. Then $\ve x_i - \ve x_i' \in \ker(\can T)$, and the only way that the difference between two vectors can be perpendicular to the longer one is if they are equal. Therefore $2p_i\begin{pmatrix}1 \\ \ve x_i\end{pmatrix}$ is the unique element of the forward light cone that $\frac12 \can \Theta$ maps to $p_i\begin{pmatrix}1 \\ \ve y_i\end{pmatrix}$, and therefore $\gamma L_{\ve b}$ times these form the only possible POVM elements for Bob. But to be a valid POVM, they must sum to the identity $\begin{pmatrix}2 \\ \ve 0\end{pmatrix}$, i.e. $\sum_i 2 p_i \begin{pmatrix}1 \\ \ve x_i\end{pmatrix} = \frac{L_{-\ve b}}{\gamma}\begin{pmatrix}2 \\ \ve 0\end{pmatrix} = 2\begin{pmatrix}1 \\ \ve b\end{pmatrix}$. Since the $\ve x_i \in \ker(\can T)^\perp$, this implies $\ve b \in \ker(\can T)^\perp$ which is equivalent to \ref{kercond}.

\underline{\ref{kercond} $\implies$ \ref{fullsteer}:} Let $\ve y_i \in \elip_A$ form a convex decomposition  $\sum_i p_i \ve y_i = \ve a$. Since $\ve y_i \in \elip_A$ we have $\ve y_i = \can {\ve a} + \can{T}\ve x_i$ where $x_i \leq 1$. Write $\ve x_i = \ve k_i + \ve c_i$ where $\ve k_i \in \ker(\can T)$ and $\ve c_i \in \ker(\can T)^\perp$. This implies $c_i \leq x_i \leq 1$ and $\ve y_i = \can{\ve a} + \can{T}\ve c_i$. So $2 p_i\begin{pmatrix}1 \\ \ve c_i\end{pmatrix}$ are in the forward light cone and map to $p_i\begin{pmatrix}1 \\ \ve y_i\end{pmatrix}$ under $\frac12 \can \Theta$. Hence $\gamma L_{\ve b}$ times these are in the forward light cone and map to $p_i\begin{pmatrix}1 \\ \ve y_i\end{pmatrix}$ under $\frac12 \Theta$. Since $\sum_i p_i \ve y_i = \ve a = \can{\ve a} + \can T\ve b$ we have $\can T\sum_i p_i \ve c_i = \can T\ve b$. By construction $\ve c_i \in \ker(\can T)^\perp$ and by assumption $\ve b \in \ker(\can T)^\perp$, and so this implies $\sum_i p_i\ve c_i = \ve b$. Then $\sum_i \gamma L_{\ve b} 2 p_i \begin{pmatrix}1 \\ \ve c_i\end{pmatrix} = 2\gamma L_{\ve b} \begin{pmatrix}1 \\ \ve b\end{pmatrix} = \begin{pmatrix}2 \\ \ve 0\end{pmatrix}$ so we have a valid POVM.

\underline{\ref{kercond} $\implies$ \ref{blochcond}:} Immediate from form of scaled down $\elip_A$ and $\ve a$ in preliminaries.

\underline{\ref{blochcond} $\implies$ \ref{kercond}:} If $\ve a$ is on the scaled down surface then $\can{\ve a} + \can T\ve x = \can{\ve a} + \can T \ve b$ where $x = b$ and $\ve x \in \ker(\can T)^\perp$. Hence $\ve x - \ve b \in \ker(\can T)$. The only way the difference between two vectors of the same length can be perpendicular to one of them is if they are the same, and so \ref{kercond} follows.

\underline{\ref{spancond} $\implies$ \ref{rangecond}:} Suppose $\sum q_i \ve v_i = 0$, $\sum q_i = 1$ (but $q$ can be negative) with $\ve v_i \in \elip_B$. Recalling that swapping parties sends $\Theta \to \Theta^T$ we see that there exists a POVM element $U_i$ in the forward light-cone with $\begin{pmatrix}1 \\ \ve v_i\end{pmatrix} = \frac12 \Theta^T U_i$. But then $\Theta^T \frac12 \sum_i q_i U_i = \sum_i q_i \frac12 \Theta^T U_i = \sum_i q_i \begin{pmatrix}1 \\ \ve v_i\end{pmatrix} = \begin{pmatrix}1 \\ \ve 0\end{pmatrix}$.

\underline{\ref{rangecond} $\implies$ \ref{spancond}:} Suppose there exists a 4-vector $U$ with $\frac12 \Theta^T U = \begin{pmatrix}1 \\ \ve 0\end{pmatrix}$. If $U = \begin{pmatrix}t \\ \ve u\end{pmatrix}$ is in the forward light cone (i.e. $t \geq u$) then $\elip_B$ itself contains the maximally mixed state and we are done. Otherwise, notice that $U_1 = \begin{pmatrix}u - t \\ 0\end{pmatrix}$ and $U_2 = \begin{pmatrix}u \\ \ve u\end{pmatrix}$ are in the forward light cone. Writing $\frac12 \Theta^T U_i$ as $q_i \begin{pmatrix} 1 \\ \ve v_i\end{pmatrix}$ we have $\ve v_i \in \elip_B$. Noting that $U = \sum_i U_i$ we have $\sum_i \frac12 \Theta^T U_i = \begin{pmatrix}1 \\ \ve 0\end{pmatrix}$, in other words $\sum_i q_i = 1$ and $\sum_i q_i \ve v_i = \ve 0$.

\underline{\ref{rangecond} $\iff$ \ref{kercond}:} For any matrix $A$, $\range(A^T) = \ker(A)^\perp$ is a theorem of linear algebra, which follows straightforwardly from the singular value decomposition.
\end{proof}

%\newpage
%\begin{thm}
%Let $\Theta$ be a two-qubit state. Let $\ve y_i \in \elip_A$ form a convex decomposition $\sum_i p_i \ve y_i = \ve a$. Then for some $q %> 0$ there exists a POVM for Bob that steers alice onto the convex decomposition $\sum_i q p_i \ve y_i + (1-q)\ve a = \ve a$.
%\end{thm}
%\begin{proof}
%Since $\ve y_i \in \elip_A$ there exists $X_i$ in the forward light cone with $\Theta X_i = p_i\begin{pmatrix}1 \\ \ve y_i\end{pmatrix}$. Let $Z = \begin{pmatrix}1 \\ \ve 0\end{pmatrix} - q\sum_i X_i$ for $0<q\leq 1$ sufficiently small that $Z$ is in the forward light cone. Then $\Theta Z =  \Theta \begin{pmatrix}1 \\ \ve 0\end{pmatrix} - q\sum_i \Theta X_i = \begin{pmatrix}1 \\ \ve a\end{pmatrix} - q\begin{pmatrix}1 \\ \ve a\end{pmatrix} = (1-q)\begin{pmatrix}1 \\ \ve a\end{pmatrix}$. Hence the $qX_i$ along with $Z$ form the required POVM.
%\end{proof}

\section{\label{Tet_proof}Steering ellipsoids in a tetrahedron in the Bloch sphere correspond to separable states}

We prove that for any $\elip_A$ inside a tetrahedron inside Alice's Bloch sphere, there is a separable state with $\rho_B = \frac12\mathbb{I}$ and $\elip_A$ is steering ellipsoid for qubit A. For the sake of brevity we drop the label A for the ellipsoid so that $\elip_A \equiv \elip$. We present the proofs for each possible dimension of $\elip$ separately, although each one is basically a slightly more involved version of the previous one. Note that in the 0 and 1 dimensional cases the requirement to fit inside a tetrahedron is trivially satisfied by any $\elip$ inside the Bloch sphere.

This result suffices to show that any state with an ellipsoid that fits inside a tetrahedron is separable by the following argument. Suppose $\rho$ has an ellipsoid for Alice (which can be degenerate) that fits inside the tetrahedron. If $b = 1$ then $\rho$ is a product state and we are done. Otherwise, apply a SLOCC operator to Bob and obtain $\can\rho$ with $\ve b = 0$, recalling that SLOCC operators cannot change a state from being entangled to separable. This will leave Alice's ellipsoid $\elip$ unchanged whilst moving her reduced state to the centre of her ellipsoid. Since by the above statement there exists a separable state with the correct ellipsoid at A and reduced states, $\can\rho$ must equal this separable state up to a choice of basis for Bob (see section \ref{sec_RecRho}), and hence must itself be separable,

In fact the separable states constructed below use a number of product states equal to the dimension of the ellipsoid plus one. Since the SLOCC operator and choice of basis for Bob do not affect the number of product states in a decomposition, we furthermore have that $\rho$ can be built using that number of product states.

\subsection{0-dimensional}
If the steering ellipsoid is a single point $\ve r$ then simply take $\rho_A$ with Bloch vector $\ve r$ and let $\rho = \rho_A \otimes \frac12 \mathbb{I}$.

\subsection{1-dimensional}
Suppose $\elip$ is a line segment from $\ve r_0$ to $\ve r_1$. Take $\rho_i$ with Bloch vectors $\ve r_i$ and let $\rho = \frac12 \sum_i \rho_i \otimes \ket{i}\bra{i}$.

\subsection{2-dimensional}
If an ellipse fits inside a tetrahedron in the Bloch sphere, it also fits inside a triangle in the Bloch sphere \cite{Mihai}. Therefore, suppose an ellipse $\elip$ fits within a triangle in the Bloch sphere whose vertices are $\{ \ve r_0, \ve r_1, \ve r_2 \}$. Without loss of generality we can take the ellipse to be tangent to each edge of the triangle, at points $\{ \ve s_i \}$ where $\ve s_i$ is on the face opposite to $\ve r_i$. Denote the centre of the ellipse by $\ve c$. Clearly there exists unique $p_i \geq 0$ such that $\sum_i p_i \ve r_i = \ve c$ and $\sum_i p_i = 1$.

By the definition of an ellipse, there is an invertible affine transformation $\A$ that maps $\elip$ to the unit circle in the $(x,z)$-plane, centred at the origin. Let $\rho_i$ have Bloch vectors $\ve r_i$ and$\ket{\psi_i}$ be such that the Bloch vector of $\ket{\psi_i}\bra{\psi_i}$ is $-\A(\ve s_i)$. We claim that the (manifestly separable) state
\begin{equation}
  \rho = \sum_i p_i \rho_i \otimes \ket{\psi_i}\bra{\psi_i}
\end{equation}
has $\rho_B = \frac{1}{2}\I$ and that Alice's steering ellipsoid for this state is $\elip$. To prove the first part, notice that the Bloch vector of $\rho_B$ is $-\sum_i p_i \A(\ve s_i)$. Since $\A$ is affine, the unit circle will be tangent to the triangle with vertices $\{\A(\ve r_i)\}$ at the points $\{\A(\ve s_i)\}$, and $\sum_i p_i \A(\ve r_i) = \A(\ve c) = \ve 0$. Hence it suffices to prove
\begin{lem} Suppose the triangle with vertices $\{ \ve v_i \}$ contains the unit circle centred at the origin, and the circle is tangent to each edge of the triangle at the points $\{ \ve t_i \}$ (where $\ve t_i$ is on the edge opposite $\ve v_i$). Fix $p_i$ by the requirements that $\sum_i p_i \ve v_i = \ve 0$ and $\sum_i p_i = 1$. Then $\sum_i p_i \ve t_i = \ve 0$.
\label{lem1}
\end{lem}
\begin{proof}
We use $\ve x$ to represent points on or within the tetrahedron using normalized barycentric co-ordinates $(x_0, x_1, x_2)$ where $\sum_i x_i = 1$ and $\vec x = \sum_i x_i \ve v_i$. Let $A_0$ be the area of the triangle with vertices $\{\ve x, \ve v_1, \ve v_2\}$, $A_1$ be the area of the triangle with vertices $\{\ve v_0, \ve x, \ve v_2\}$ and similarly for $A_2$. Let $A$ be the area of the original triangle (notice $A = \sum_i A_i$). Then $x_i = A_i / A$. By definition the barycentric co-ordinates of the origin are $(p_0, p_1, p_2)$.

Let $L_i$ be the length of the edge opposite $\ve v_i$, and let $L = \sum_i L_i$. By using that the area of a triangle = $\frac12$ (base) $\times$ (perpendicular height)
and noting that by the tangency assumption the relevant triangles have a perpendicular height of 1, we obtain that $p_i = L_i / L$.

\begin{figure}
\includegraphics{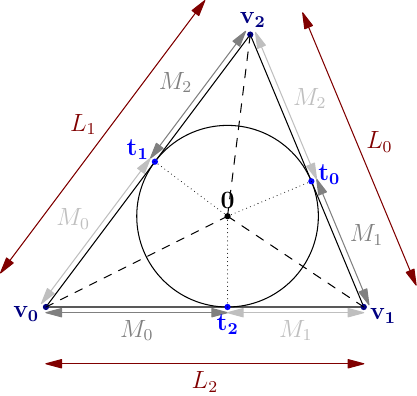}
\caption{The various quantities used in proving Lemma~\ref{lem1}. The dashed lines form the three triangles used to show $p_i = L_i/L$, the dotted lines indicate their perpendicular heights (which are equal to the radius of the circle: 1).}
\label{lemfig}
\end{figure}

Let $M_0^{(1)} = \abs{\ve v_0 - \ve t_1}$, $M_0^{(2)} = \abs{\ve v_0 - \ve t_2}$. In fact $M_0^{(1)} = M_0^{(2)}$ because they are both the unique length defined by the requirement of being from a fixed point to a point on the circle such that the line between them is tangent to the sphere, so we can write this length simply as $M_0$. Define the other two $M_i$ by a similar argument. All this is illustrated in Figure~\ref{lemfig}. Notice that
\begin{align}
L_0 &= M_{1} + M_{2},\\
L_1 &= M_{0} + M_{2},\\
L_2 &= M_{0} + M_{1},\\
\end{align}
The barycentric co-ordinates of $\ve t_0, \ve t_1$ and $\ve t_2$ can now be calculated as
\begin{equation}
  (0, M_2, M_1)/L_0,
\end{equation}
\begin{equation}
  (M_2, 0, M_0)/L_1,
\end{equation}
and
\begin{equation}
  (M_1, M_0, 0)/L_2,
\end{equation}
respectively. Using $p_i = L_i / L$ and the fact that barycentric coordinates respect convex combinations the required result is now immediate.
\end{proof}
Suppose Bob projects his qubit onto $\ket{\psi}$ and the orthogonal state. Since $\rho_B = \frac{1}2 \I$ he will obtain each outcome with probability $\frac12$. Therefore if he obtains the $\ket{\psi}$ outcome then Alice's state will be
\begin{eqnarray}
  \rho_{A}(\ket{\psi}) &:=& \frac{\tr_B(\rho(\I \otimes \ket{\psi}\bra{\psi}))}{\tr(\rho_B\ket{\psi}\bra{\psi})}  \\
  &=& \frac{\sum_i p_i \rho_i \abs{\braket{\psi_i}{\psi}}^2}{\frac12} \\
  &=& 2\sum_i p_i \rho_i \abs{\braket{\psi_i}{\psi}}^2
\end{eqnarray}
Recalling that the Bloch vector of $\ket{\psi_i}\bra{\psi_i}$ is $-\A(\ve s_i)$, then if $\ket{\psi}\bra{\psi}$ has Bloch vector $\ve r$ then the Bloch vector of $\rho_A(\ket{\psi})$ will be
\begin{equation}
f(\ve r) := 2 \sum_i p_i \ve r_i \frac{1 - \ve r \cdot \A(\ve s_i)}2 = \sum_i p_i \ve r_i \left(1 - \ve r \cdot \A(\ve s_i)\right).
\end{equation}
Let us extend this expression to all $\ve r$ to define an affine function $f$. The statement that Alice's steering ellipsoid is $\elip$ is equivalent to the statement that $\elip$ is the image of the unit sphere under $f$. Since all the $\A(\ve s_i)$ are in the $(x,z)$-plane, we have $f\left( (0,1,0) \right) = f(\ve0)$, i.e. we can think of $f$ as first projecting onto the $(x,z)$-plane and then applying some affine transformation. The image of the unit sphere under that projection is the unit disc, and so it suffices to check that $\E$ is the image of the unit circle under $f$. Define $g(\ve r) = \A(f(\ve r))$. Since $\A$ is invertible and maps $\elip$ to unit circle it suffices to prove that $g$ is the identity on the $(x,z)$-plane. Since $g$ is the composition of two affine functions it is also affine. By the definition of the $p_i$, $g(\ve 0) = \A(\ve c) = \ve 0$ so $g$ is in fact linear. Hence it suffices to check that $g(\ve u_j) = \ve u_j$ for some spanning set of vectors $\{\ve u_j\}$. Since the triangle cannot be degenerate, its vertex set $\{ \ve r_j \}$ span some plane. Since $\A$ is invertible, $\{ \A(\ve r_j) \}$ must span the $(x,z)$ plane. For $i \neq j$, $\{\ve 0, \A(\ve s_i), \A(\ve r_j)\}$ form a right-angle triangle, and $\abs{\A(\ve s_i)} = 1$. Therefore $\A(\ve r_j) \cdot \A(\ve s_i) = 1$ whenever $i \neq j$. But $\sum_i p_i(1 - \ve r \cdot \A(\ve s_i)) = \sum_i p_i - \ve r\cdot\left(\sum_i p_i \A(\ve s_i)\right) = 1 - \ve r \cdot \ve 0 = 1$ (the penultimate equality is from Lemma~\ref{lem1}). Hence $p_i(1 - \A(\ve r_j)\cdot\A(\ve s_i)) = \delta_{ij}$ and we are done.

\subsection{3-dimensional}

Suppose an ellipsoid $\elip$ fits within a tetrahedron in the Bloch sphere whose vertices are $\{ \ve r_0, \ve r_1, \ve r_2, \ve r_3 \}$. Without loss of generality we can take the ellipsoid to be tangent to each face of the tetrahedron, at points $\{ \ve s_i \}$ where $\ve s_i$ is on the face opposite to $\ve r_i$. Denote the centre of the ellipsoid by $\ve c$. Clearly there exists unique $p_i \geq 0$ such that $\sum_i p_i \ve r_i = \ve c$ and $\sum_i p_i = 1$.

By the definition of an ellipsoid, there is an invertible affine transformation $\A$ that maps $\elip$ to the unit sphere centred at the origin. Let $\rho_i$ have Bloch vectors $\ve r_i$ and $\ket{\psi_i}$ be such that the Bloch vector of $\ket{\psi_i}\bra{\psi_i}$ is $-\A(\ve s_i)$. We claim that the (manifestly separable) state
\begin{equation}
  \rho = \sum_i p_i \rho_i \otimes \ket{\psi_i}\bra{\psi_i}
\end{equation}
has $\rho_B = \frac{1}{2}\I$ and that Alice's steering ellipsoid for this state is $\elip$. To prove the first part, notice that the Bloch vector of $\rho_B$ is $-\sum_i p_i \A(\ve s_i)$. Since $\A$ is affine, the unit sphere will be tangent to the tetrahedron with vertices $\{\A(\ve r_i)\}$ at the points $\{\A(\ve s_i)\}$, and $\sum_i p_i \A(\ve r_i) = \A(\ve c) = \ve 0$. Hence it suffices to prove the following 3-dimensional analogue to Lemma~\ref{lem1}:
\begin{lem} Suppose the tetrahedron with vertices $\{ \ve v_i \}$ contains the unit sphere centred at the origin, and the sphere is tangent each face of the tetrahedron at the points $\{ \ve t_i \}$ (where $\ve t_i$ is on the face opposite $\ve v_i$). Fix $p_i$ by the requirements that $\sum_i p_i \ve v_i = \ve 0$ and $\sum_i p_i = 1$. Then $\sum_i p_i \ve t_i = \ve 0$.
\label{lem2}
\end{lem}
\begin{proof}
We use $\ve x$ to represent points on or within the tetrahedron using normalized barycentric co-ordinates $(x_0, x_1, x_2, x_3)$ where $\sum_i x_i = 1$ and $\vec x = \sum_i x_i \ve v_i$. Let $V_0$ be the volume of the tetrahedron with vertices $\{\ve x, \ve v_1, \ve v_2, \ve v_3\}$, $V_1$ be the volume of the tetrahedron with vertices $\{\ve v_0, \ve x, \ve v_2, \ve v_3\}$ and so on. Let $V$ be the volume of the original tetrahedron (notice $V = \sum_i V_i$). Then $x_i = V_i / V$. By definition the barycentric co-ordinates of the origin are $(p_0, p_1, p_2, p_3)$.

Let $A_i$ be the area of the face opposite $\ve v_i$, and let $A = \sum_i A_i$. By using that the volume of a tetrahedron = $\frac13$ (area of base) $\times$ (perpendicular height)
%\begin{widetext}
%\begin{equation}
%\text{volume of a tetrahedron} = \frac13 (\text{area of base}) \times (\text{perpendicular height})
%\end{equation}
%\end{widetext}
and noting that by the tangency assumption the relevant tetrahedra have a perpendicular height of 1, we obtain that $p_i = A_i / A$.

Let $A_{23}^{(0)}$ be the area of the triangle with vertices $\{\ve v_2, \ve v_3, \ve t_0\}$. Let $A_{23}^{(1)}$ be the area of the triangle with vertices $\{\ve v_2, \ve v_3, \ve t_1\}$. Now we have that $\abs{\ve v_2 - \ve t_0} = \abs{\ve v_2 - \ve t_1}$ because they are both the unique length defined by the requirement of being from a fixed point to a point on the sphere such that the line between them is tangent to the sphere. Similarly $\abs{\ve v_3 - \ve t_0} = \abs{\ve v_3 - \ve t_1}$. Hence the two triangles are congruent and we can simply write their areas as $A_{23}$. Define the other five $A_{ij}$ by a similar argument. Notice that
\begin{align}
A_0 &= A_{12} + A_{13} + A_{23},\\
A_1 &= A_{02} + A_{03} + A_{23},\\
A_2 &= A_{01} + A_{03} + A_{13},\\
A_3 &= A_{01} + A_{02} + A_{12}.
\end{align}
The barycentric co-ordinates of $\ve t_0, \ve t_1, \ve t_2$ and $\ve t_3$ can now be calculated as
\begin{equation}
  (0, A_{23}, A_{13}, A_{12})/A_0,
\end{equation}
\begin{equation}
  (A_{23}, 0, A_{03}, A_{02})/A_1,
\end{equation}
\begin{equation}
  (A_{13}, A_{03}, 0, A_{01})/A_2,
\end{equation}
and
\begin{equation}
  (A_{12}, A_{02}, A_{01}, 0)/A_3
\end{equation}
respectively. Using $p_i = A_i / A$ and the fact that barycentric coordinates respect convex combinations the required result is now immediate.
\end{proof}
As in the 2-dimensional case we find that if Bob projects onto the state with Bloch vector $\ve r$ then Alice's Bloch vector is
\begin{equation}
f(\ve r) := \sum_i p_i \ve r_i \left(1 - \ve r \cdot \A(\ve s_i)\right).
\end{equation}
Let us extend this expression to all $\ve r$ to define an affine function $f$. The statement that Alice's steering ellipsoid is $\elip$ is equivalent to the statement that $\elip$ is the image of the unit sphere under $f$. Define $g(\ve r) = \A(f(\ve r))$. Since $\A$ is invertible and maps $\elip$ to unit sphere it suffices to prove that $g$ is the identity. Since $g$ is the composition of two affine functions it is also affine. By the definition of the $p_i$, $g(\ve 0) = \A(\ve c) = \ve 0$ so $g$ is in fact linear. Hence it suffices to check that $g(\ve u_j) = \ve u_j$ for some spanning set of vectors $\{\ve u_j\}$. Since the tetrahedron cannot be degenerate, its vertex set $\{ \ve r_j \}$ must be spanning. Since $\A$ is invertible, $\{ \A(\ve r_j) \}$ is also spanning. As in the 2-dimensional case, for $i \neq j$, $\{\ve 0, \A(\ve s_i), \A(\ve r_j)\}$ form a right-angle triangle, and $\abs{\A(\ve s_i)} = 1$. Therefore $\A(\ve r_j) \cdot \A(\ve s_i) = 1$ whenever $i \neq j$. But $\sum_i p_i(1 - \ve r \cdot \A(\ve s_i)) = \sum_i p_i - \ve r\cdot\left(\sum_i p_i \A(\ve s_i)\right) = 1 - \ve r \cdot \ve 0 = 1$ (the penultimate equality is from Lemma~\ref{lem2}). Hence $p_i(1 - \A(\ve r_j)\cdot\A(\ve s_i)) = \delta_{ij}$ and we are done.

\section{Discord and steering ellipsoids}\label{sec_discord}

Below we outline the condition for zero discord for Alice from either her or Bob's ellipsoid.

\begin{quote}\textit{A state has zero discord for Alice iff her ellipsoid degenerates to a segment of a diameter.}% *** NEED TO ADD THIS IN ***
\end{quote}

The ``only if'' part: A general zero discord state for Alice $\rho = p \ket{e}\bra{e} \otimes \rho_0 + (1-p)\ket{\bar{e}}\bra{\bar{e}} \otimes \rho_1$ has $\langle e |\bar{e} \rangle = 0$ and
\begin{eqnarray}
\ve {a} &=& t \ve{e}\\
\ve {b} &=& \ve {x} \\
T &=& \ve{e y}^T
\end{eqnarray}
where $t = 2p-1$, $\ve{e} = \bra{e} \boldsymbol{\sigma} \ket{e}$ and $\ve {x} = \tr \left[ \left(p \rho_0 + (1-p)\rho_1\right)  \boldsymbol{\sigma} \right]$, $\ve {y} = \tr \left[ \left(p \rho_0 - (1-p)\rho_1\right)  \boldsymbol{\sigma} \right]$ \cite{VlatkoDiscord}.

Alice's steering ellipsoid $\mathcal{E}_A$ has centre
%\begin{equation}
$\ve{c}_A = \left(\frac{t- \ve{x}\cdot\ve{y}}{1-x^2}\right)\ve{e}$
%\end{equation}
and matrix
%\begin{align}
$Q_A = s_A^2 \ve {e} \ve{e}^T$ with
%\end{align}
\begin{equation}
s_A^2 = \frac{1}{1-x^2} \left[ \left( \ve{y}-t\ve{x} \right)^T \left( \I + \frac{\ve{x}\ve{x}^T}{1-x^2}\right)  \left( \ve{y}-t\ve{x} \right)\right]
\end{equation}
So $\mathcal{E}_A$ is a segment of the diameter.

The ``if'' part: suppose we are given $\mathcal{E}_A$, a segment of the diameter. Denote the states endpoints of the ellipsoid $\rho_0$ and $\rho_1$. Alice's state can always be decomposed as $\rho_A = q \rho_0 + (1-q) \rho_1$. However since all the Bloch vectors of $\rho_A, \rho_0, \rho_1$ are collinear they will eigendecompose into the same pair of orthogonal states, call them $\ket{\psi},\ket{\bar{\psi}}$. Writing $\rho_i = p_i \ket{\psi} \bra{\psi} + (1-p_i)\ket{\bar{\psi}} \bra{\bar{\psi}}$, for $i=0,1$ then $\rho_A = p \ket{\psi} \bra{\psi} + (1-p) \ket{\bar{\psi}} \bra{\bar{\psi}}$ with $p = qp_0 + (1-q)p_1$. Then the joint state
\begin{align}
\rho = p \ket{\psi} \bra{\psi}\otimes \beta_0 + (1-p) \ket{\bar{\psi}} \bra{\bar{\psi}}\otimes\beta_1
\end{align}
is a zero discord state for Alice with the correct $\mathcal{E}_A$ and $\rho_A$  for any mixed states on Bob's side $\beta_0,\beta_1$.

\begin{quote}\textit{There is zero discord for Alice iff Bob's ellipsoid degenerates to a line segment and the length of Alice's Bloch vector is equal to the distance from the centre of Bob's ellipsoid to his Bloch vector divided by the radius of his ellipsoid.}
\end{quote}

The ``only if'' part can easily be checked: it requires $a = \frac{|\boldsymbol{c}_B-\boldsymbol{b}|}{s_B}$. Since, after some algebra,
\begin{eqnarray}
\ve{c}_B &=& \frac{\ve{x} - t \ve{y}}{1-t^2} \\
Q_B &=& \frac{1}{(1-t^2)^2}(\ve{y} - t\ve{x})(\ve{y} - t\ve{x})^T
\end{eqnarray}
then $|\boldsymbol{c}_B-\boldsymbol{b}| = \frac{t |\ve{y} - t\ve{x}|}{1-t^2} = as_B$.

For the ``if'' part, let $\rho_0$ and $\rho_1$ be the endpoints of Bob's ellipsoid, and let the state corresponding to Alice's Bloch vector have eigendecomposition $\rho_A = p_0 \ket{\psi_0}\bra{\psi_0} + p_1 \ket{\psi_1}\bra{\psi_1}$. Then the joint state
\begin{equation}
\rho = p_0  \ket{\psi_0}\bra{\psi_0} \otimes \rho_0 + p_1\ket{\psi_1}\bra{\psi_1}  \otimes  \rho_1
\end{equation}
has zero discord for Alice, the correct Bloch vector for Alice and the correct ellipsoid for Bob. If necessary swapping $\rho_0$ and $\rho_1$, it also has the right Bloch vector for Bob. Bob's ellipsoid $\mathcal{E}_B$ is invariant under local unitaries on Alice's qubit, so Alice and Bob's actual state is therefore equivalent to $\rho$ up to this transformation, which preserves discord.

\section{Volume formula for the steering ellipsoid}\label{sec_volform}

%We now provide a compact and useful expression for the volume of the steering ellipsoid, which allows a simple inference as to the type of steering a given state $\rho$ provides.
The volume $V$ of any ellipsoid is proportional to the product of its semiaxes: $V=\frac{4\pi}{3}s_1s_2s_3$. Therefore the ellipsoid $\mathcal{E}_A$ has volume $V_A = \frac{4\pi}{3}\sqrt{\det Q_A}$, where $Q_A$ is the matrix whose eigenvalues are the squares of the semiaxes, equation \eqref{QA} in the main text. Working out the determinant of $Q_A$ gives the volume terms of the components of the state $\rho$
\begin{eqnarray}
\label{VolumeRaw}
V_A&=&\frac{4\pi }{3}\frac{\left\vert \det(T- \boldsymbol{a}\boldsymbol{b}^T) \right\vert }{\left( 1-b^{2}\right) ^{2}} \\
&=&\frac{4\pi }{3}\frac{\left\vert \det\Theta\right\vert }{\left( 1-b^{2}\right) ^{2}}
\end{eqnarray}
where in the second line we have used the fact that $\det(T- \boldsymbol{a}\boldsymbol{b}^T) = \mathrm{det}\Theta$.

To express this in terms of the density matrix $\rho$, we use the equation \cite{FrankFilter}
\begin{equation}
\Theta = 2 \Upsilon \rho^R \Upsilon^T
\label{Theta_rho_reln}
\end{equation}
where the unitary matrix $\Upsilon$ is given in equation \eqref{Seagull}
and $R$ denotes a reshuffling operation: if $\rho = \sum_{i,j=0}^1 \rho_{ij;kl}\ket{ij}\bra{kl}$ then $\rho^R=\sum_{i,j=0}^1 \rho_{ik;jl}\ket{ij}\bra{kl}$. We also require a curious relation that holds for any 4$\times$4 or 9$\times$9 matrix $M$:
\begin{equation}
\det M = \det M^{T_B} - \det( M^{T_B})^R.
\label{dets_reln}
\end{equation}
Consider what happens to $\Theta$ as $\rho \rightarrow \rho^{T_B}$. The basis of $\Theta$ is $\{ \sigma_\mu \otimes \sigma_\nu\}$, with $\mu,\nu = 0,1,2,3$. The partial transpose on B changes the sign of $\sigma_2$, i.e. the $T_B$ operation sends the components $\Theta_{\mu 2} \rightarrow - \Theta_{\mu 2} $, this can be expressed as $\Theta \rightarrow \Theta \Omega$, where $\Omega = \diag(1,1,-1,1)$.
Using this, and the inverted version of equation \eqref{Theta_rho_reln}, $\rho^R = \frac{1}{2} \Upsilon^{\dag}\Theta \Upsilon^{*} $, it follows that
\begin{equation}
\left(\rho^{T_B} \right)^R = \frac{1}{2} \Upsilon^{\dag}\Theta \Omega \Upsilon^{*}
\end{equation}
and so
\begin{eqnarray}
\det\left(\rho^{T_B} \right)^R &=& \frac{1}{2^4} (\det \Upsilon^{*})^2 \det\Theta \det\Omega \\
&=& \frac{1}{2^4}  \det\Theta
\end{eqnarray}
since $\det \Upsilon = i$ and $ \det\Omega =-1$.
Inserting this into equation \eqref{dets_reln} with $M=\rho$ we discover the relation between the determinants
\begin{equation}
\det\Theta= 16 (\det \rho^{T_B} -\det \rho)
\label{det_Theta_rho}
\end{equation}

This can now be substituted into the volume equation (\ref{VolumeRaw}) to yield the volume of $\mathcal{E}_A$ in terms of the state $\rho$ as
\begin{equation}
V_A = \frac{64\pi }{3}\frac{\left\vert \det  \rho^{T_B}
-\det  \rho \right\vert }{\left( 1-b^{2}\right) ^{2}}.
\end{equation}

The volume for ellipsoid B $$ V_B \propto \sqrt{\det Q_B} = \frac{\left\vert \det(T^T- \boldsymbol{b}\boldsymbol{a}^T) \right\vert }{\left( 1-a^{2}\right) ^{2}}, $$ but $ \det(T^T- \boldsymbol{b}\boldsymbol{a}^T) =\det \Theta^T = \det\Theta$ hence
$V_B = \frac{(1-b^2)^2}{(1-a^2)^2}V_A$.
\newline

\section{The steering ellipsoid zoo}
In this section we illustrate the main types of ellipsoid.
\subsection{Entangled states}
For every pure entangled state the ellipsoid coincides with the Bloch sphere.
When the state is mixed and entangled, the ellipsoid does not satisfy the tetrahedral condition because, loosely speaking,
the ellipsoid is either \emph{too big} (with volume $V \ge V_\star$, where $V_\star$ is the largest separable volume $4\pi/81$) or \emph{too near} (large $c$) to the surface of the Bloch sphere, see figure \ref{entangled}. Every entangled state is completely steerable.
\begin{figure}[h]
\includegraphics[width=3cm]{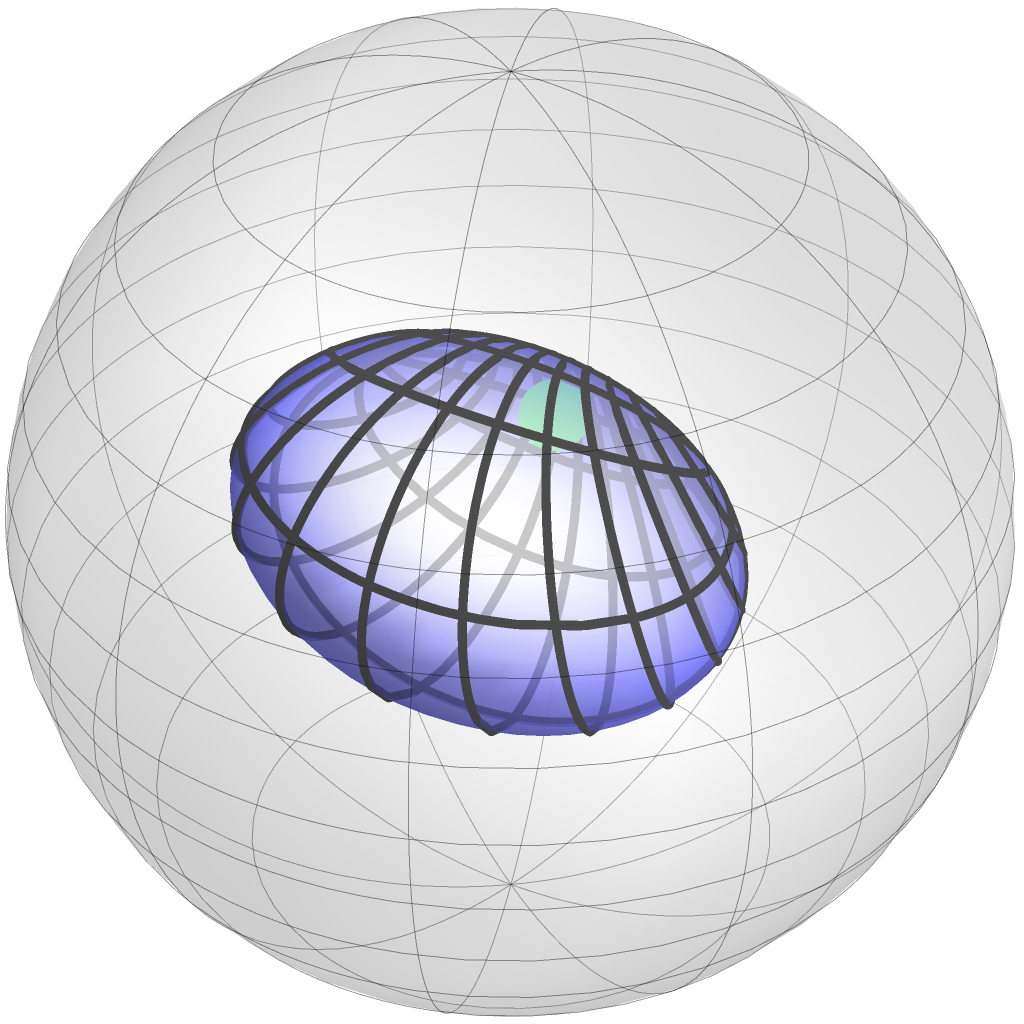}
\includegraphics[width=3cm]{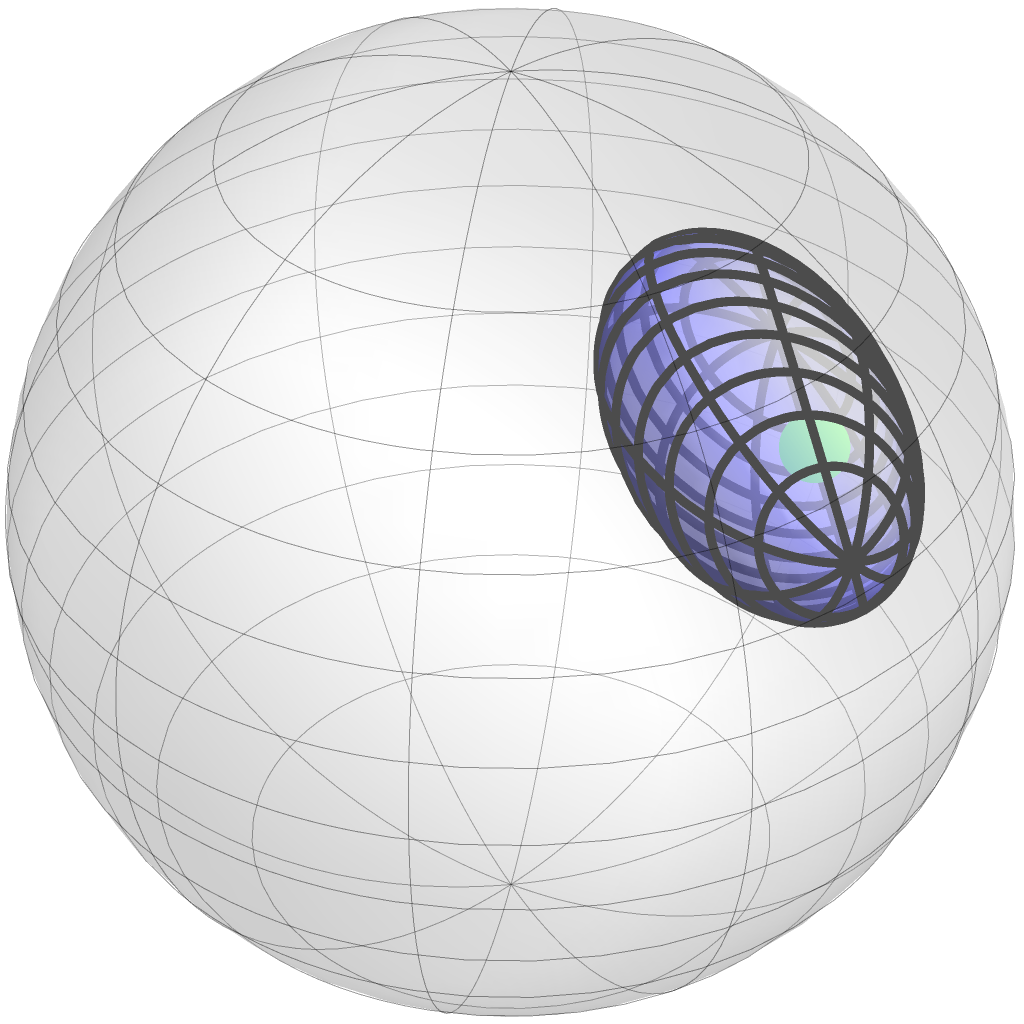}
\caption{A generic entangled state $  \rho$: both ellipsoids $\E_A$ and $\E_B$ are always full rank and neither can be inscribed within a tetrahedron within the Bloch sphere.}
\label{entangled}
\end{figure}

\subsection{Separable states with full-dimensional ellipsoids}

Separable states admit a convex decomposition in terms of product states, and have ``more classical" correlations. Steering is still possible, however the steering ellipsoids necessarily obey the tetrahedral condition, as in figure \ref{separable}.

If the state has a three dimensional $\E_A$ then it has non-zero obesity and non-zero discord, and furthermore, it can be written as a mixture of just four product states. Such states are also completely steerable.
\begin{figure}[h]
\includegraphics[width=3cm]{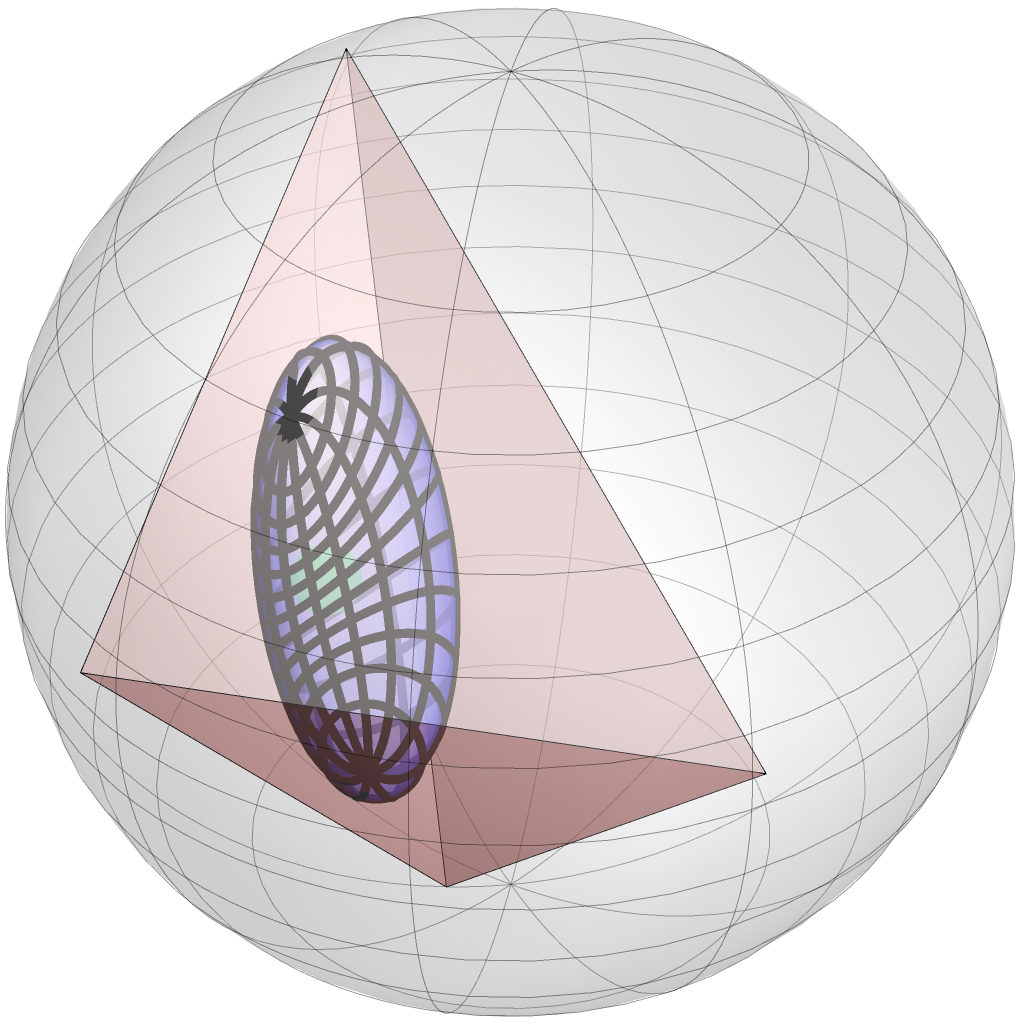}
\includegraphics[width=3cm]{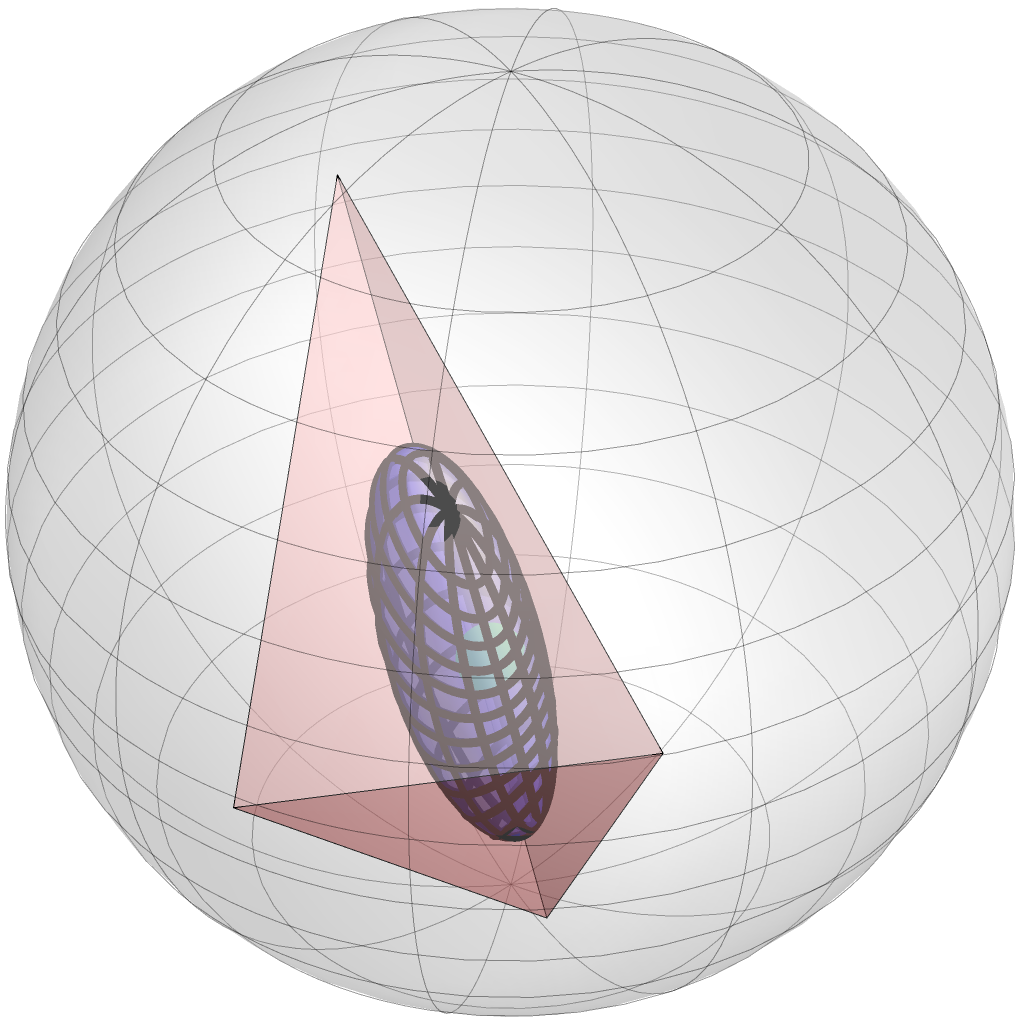}
\caption{A generic separable state $  \rho$, where both ellipsoids $\E_A$ and $\E_B$ are full rank and fit inside a tetrahedron.}
\label{separable}
\end{figure}

\begin{figure}[h]
\includegraphics[width=4cm]{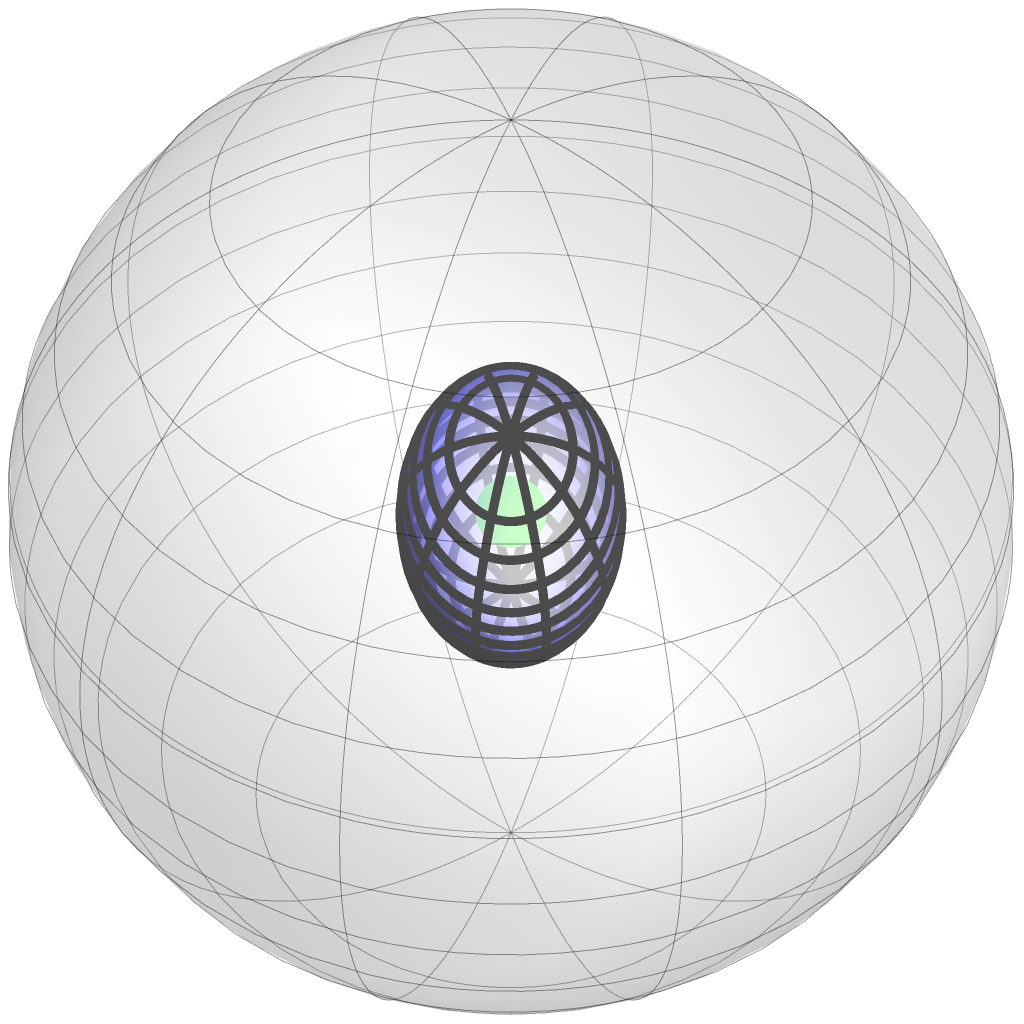}
\caption{A Bell-diagonal state: the ellipsoid is centred at the origin and its semiaxes are given by the three singular values of $T$. The vector of these singular values $\boldsymbol{t} = (t_1,t_2,t_3)$ lives in a tetrahedron with vertices at (1,1,-1),(1,-1,1),(-1,1,1),(-1,-1,-1), and when it is inside an octahedron inside of this tetrahedron, then the state is necessarily separable. This defines the set of ellipsoids that fit inside the nested tetrahedron (these are not the same tetrahedra).}
\label{bell}
\end{figure}
\begin{figure}[h]
\includegraphics[width=3cm]{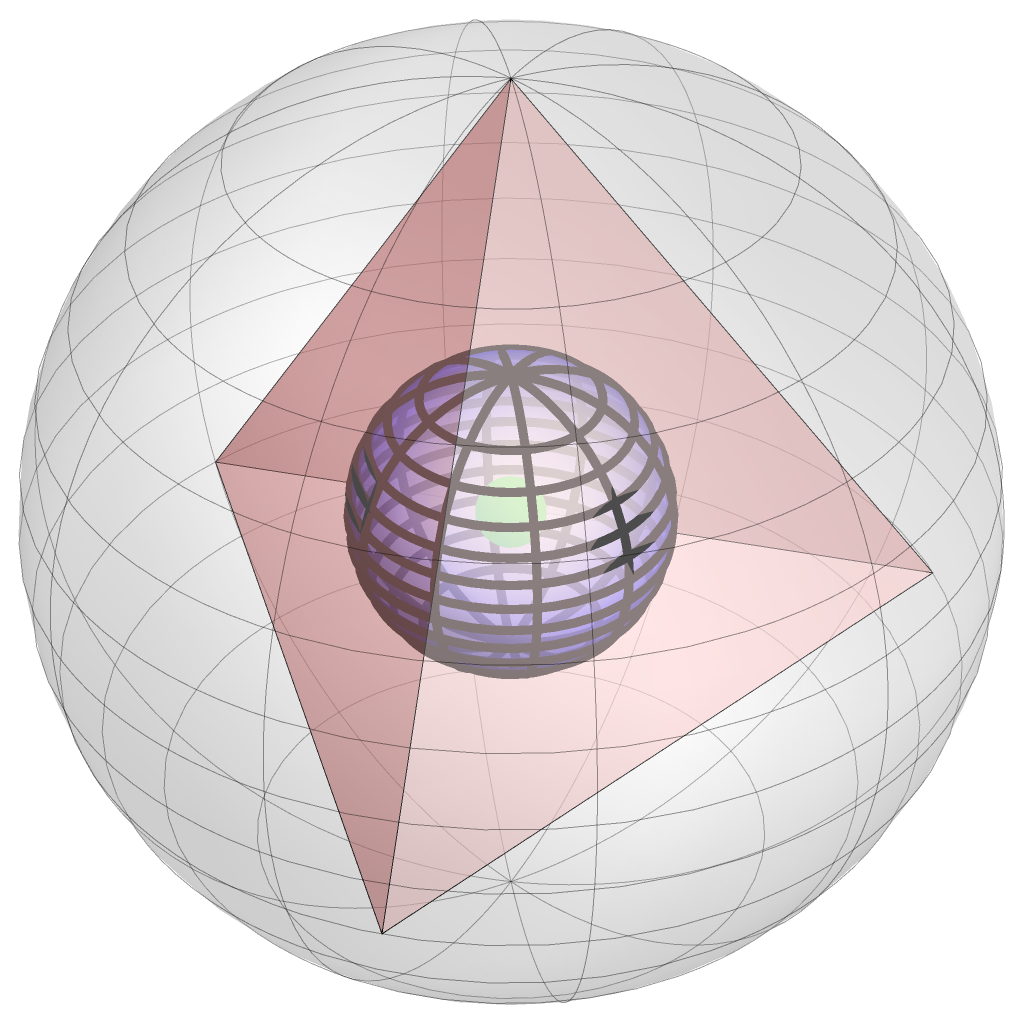}
\caption{A Werner state. The steering ellipsoid $\E_A$ is a sphere centred at the origin. The ellipsoid fits inside a tetrahedron when its radius is less than $\frac{1}{3}$ and thus the state is separable.}
\label{werner}
\end{figure}
%\newpage
\subsection{Steering pancakes}
The set of states that Bob can steer Alice to may become degenerate, and form a \emph{two-dimensional} set. This ``steering pancake" will not only fit inside a tetrahedron, but will fit within a \emph{triangle} that is inscribed within the Bloch sphere as shown in \ref{pancake}.
\begin{figure}[h]
\includegraphics[width=3cm]{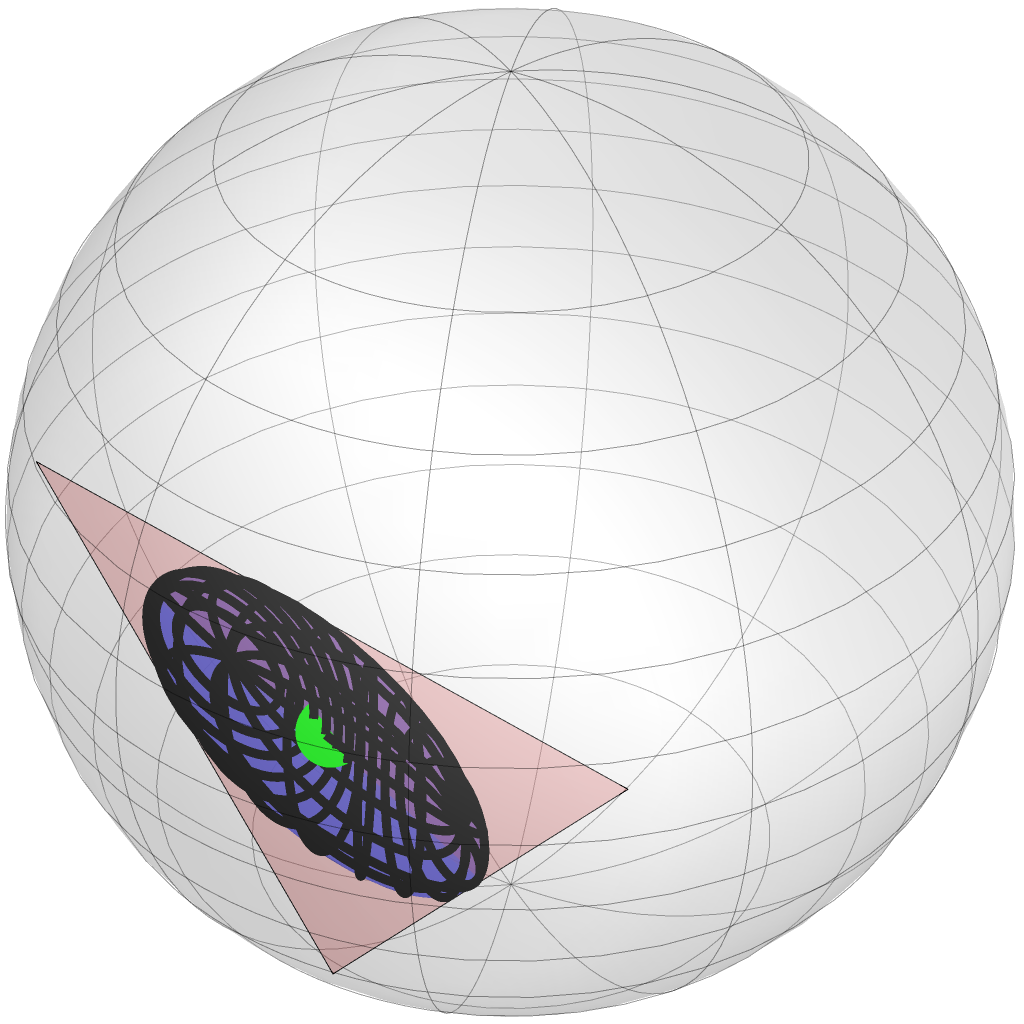}
\includegraphics[width=3cm]{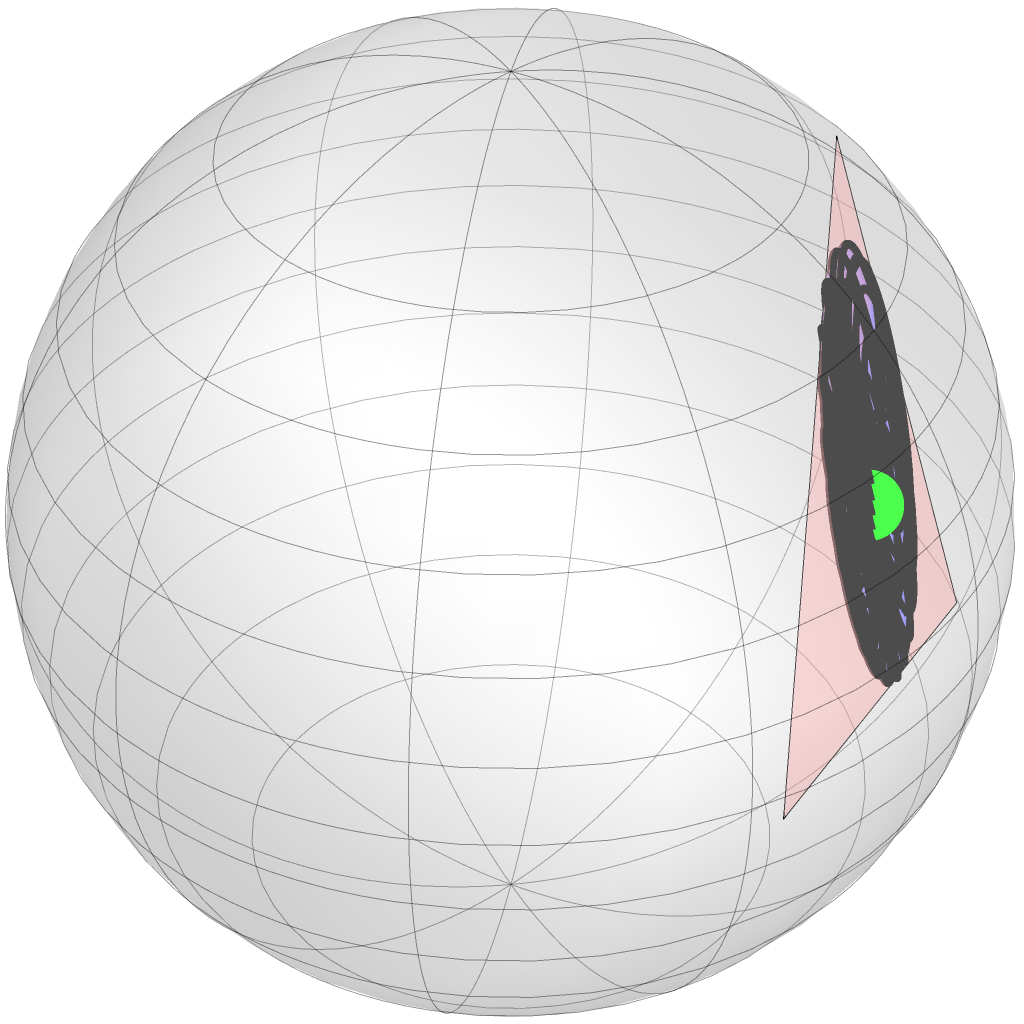}
\caption{A generic separable state $  \rho$, where both ellipsoids $\E_A$ and $\E_B$ are steering pancakes.}
\label{pancake}
\end{figure}
Recall that we have a novel feature for some steering pancakes (and steering needles) of \emph{incomplete} steering. For steering pancakes we have complete steering of qubit A if and only if the affine span of $\E_B$ contains the origin of the Bloch sphere.
%\newpage
\subsection{Steering needles}
The steering can become even more degenerate, and the steering set collapses to a one-dimensional line segment, or ``steering needle". These states include perfectly classical (doubly zero-discord states) with needles being radial (figure \ref{needle}), but also includes non-zero discord states for which either one or both of the steering needles $\E_A$ and $\E_B$ is not radial (figure \ref{needle2}). Being radial is indicated by dashed lines on the figures, which depict diameters.
\begin{figure}[h]
\includegraphics[width=3cm]{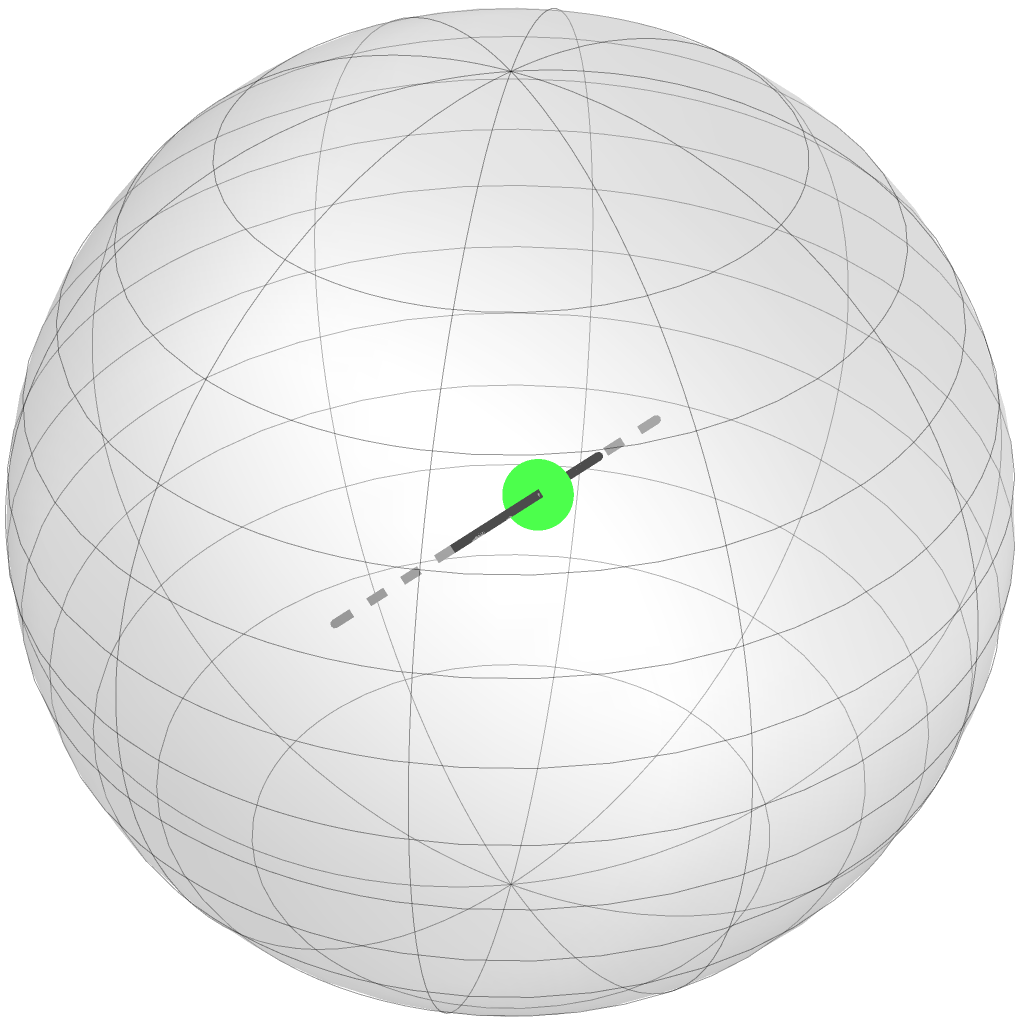}
\includegraphics[width=3cm]{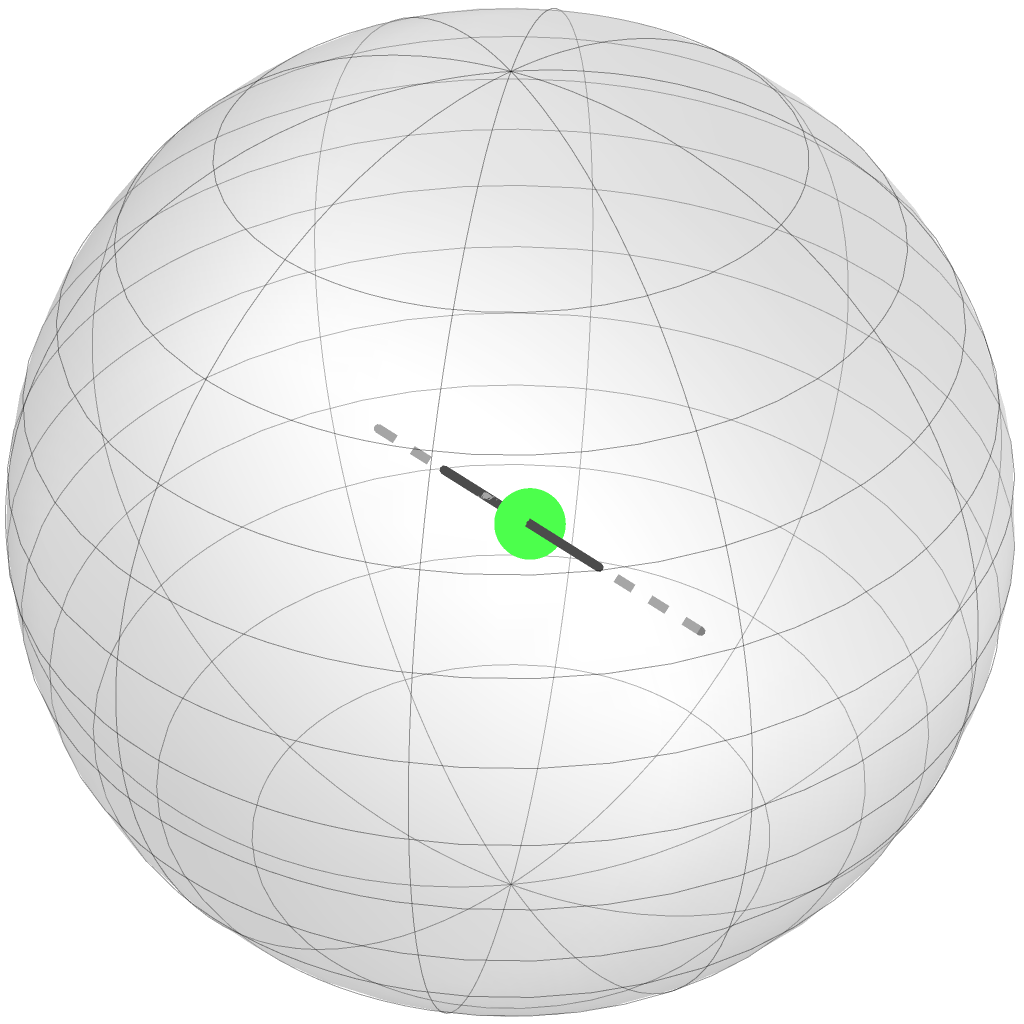}
\caption{A doubly zero discord state, where both $\E_A$ and $\E_B$ are radial line segments.}
\label{needle}
\end{figure}
\begin{figure}[h]
\includegraphics[width=3cm]{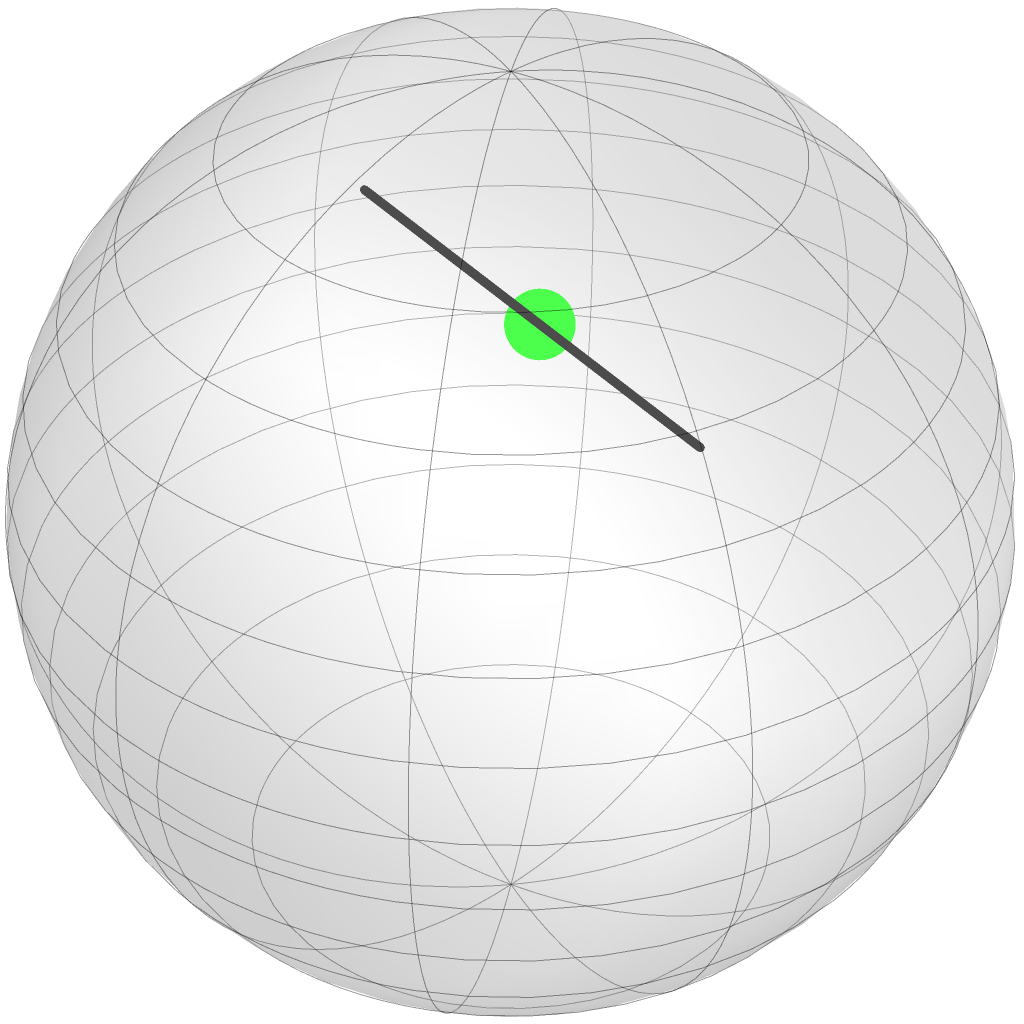}
\includegraphics[width=3cm]{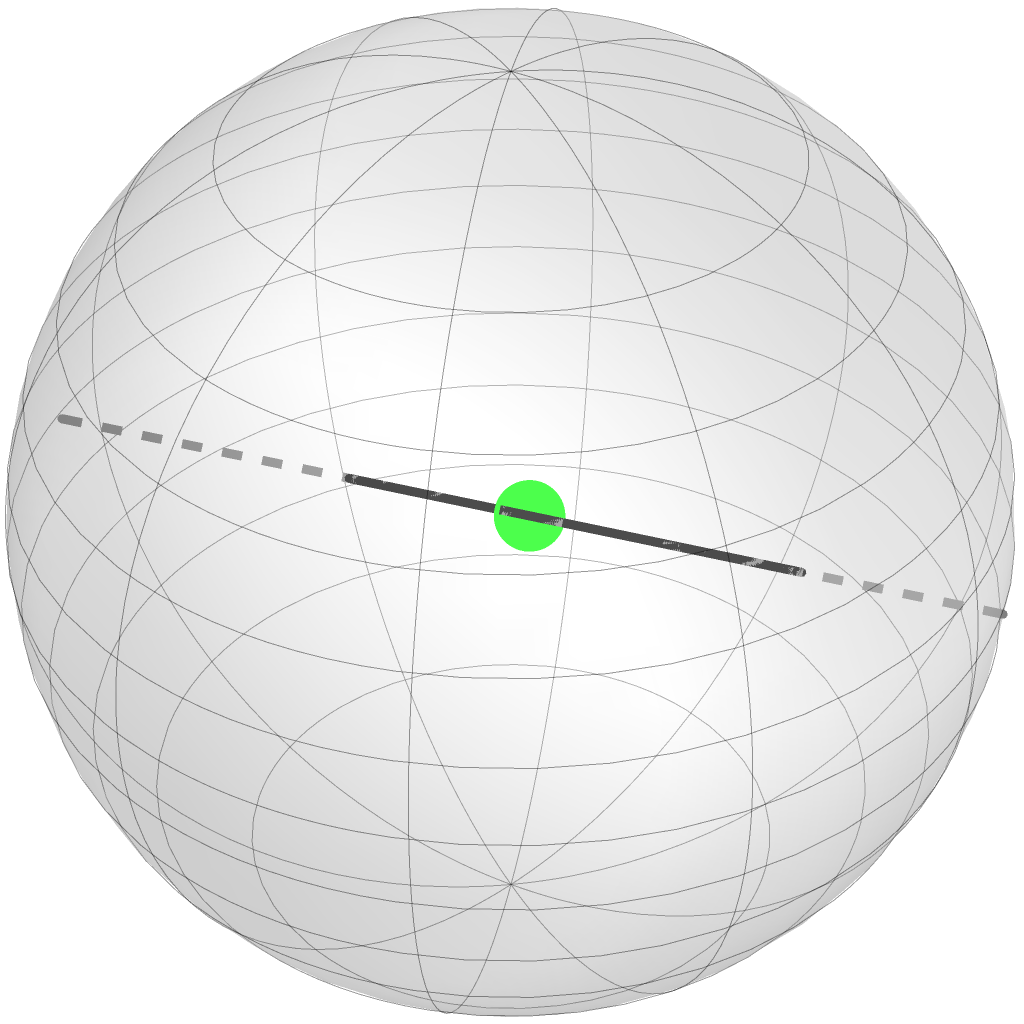}
\caption{A state with zero discord at $B$, but \emph{non-zero} discord at $A$. We have $\E_B$ being a radial line segment, but $\E_A$ is not a radial line segment.}
\label{needle2}
\end{figure}

\section{Beyond two qubits: Extensions to higher dimensional quantum systems}\label{sec_extend}
It is relatively straightforward to outline which aspects of our work naturally generalise beyond the two qubit case, and also where new directions lie. There are many interesting open questions in this regard, and the two simplest scenarios to consider are (a) when Alice holds a $d>2$ dimensional quantum system, and (b) when we have a system of $N$ qubits distributed in space and in some large entangled multipartite state. We briefly consider each case in turn.
\subsection{$d \times 2$ systems}
If Alice has a $d$-dimensional system (with $d>2$) while Bob possesses a qubit, then the set of states that Bob can steer Alice to will still form an ellipsoid. This can be seen since the derivation in section \ref{sec_SteerEll_Deriv} easily extends to accommodate this scenario. However, now the ellipsoid is embedded in the space of qu$d$it ``Bloch'' vectors at Alice's side, which has $d^2 - 1$ real dimensions. This is a much more intricate structure than the Bloch sphere, see for example \cite{qutrit}. Our results fall into three categories.

Firstly, some of our results and their proofs carry over straightforwardly to this setting: the embedded ellipsoid and local Bloch vectors will still suffice as a faithful representation of the bipartite density matrix as we are still able to reconstruct the state up to a choice of basis for Bob. Furthermore, fitting inside a tetrahedron inside the qu$d$it state space will still suffice for the state to be separable.

Secondly, some of our results carry over with only slight modification, for example states that have zero discord for Alice would correspond to ellipses inside a triangle whose vertices are orthogonal states for $d=3$, and ellipsoids inside a tetrahedron with vertices on orthogonal states for $d=4$.

For some of our results, the third category, the situation is unclear and it is here that perhaps the most interesting features are to be found. For instance the necessity of fitting inside a tetrahedron for separability followed simply from the fact that any two-qubit separable state can be written using at most four product states \cite{Sanpera}. However,for the case of $3 \times 2$ dimensional separable states, it is known that this can require as many as 6 \emph{pure} product states \cite{qubitqudit}, to the best of our knowledge it is an open question how few product states are required if mixed states can be used. This question has an interesting geometric aspect as motivated from the steering ellipsoid perspective.

In addition, several of our results are based on the PPT criterion, and so may extend to $d=3$, where the criterion is still necessary and sufficient, and may admit application to the set of PPT states for $d>3$, where it is now a strict superset of the separable states.

\subsection{$2 \times 2 \times 2 \dotsb 2$ systems}
Another natural extension of our work is to a multipartite system composed of entangled qubits. At the simplest level, given such a multipartite qubit state, we can consider the steering ellipsoid for any party to steer any other, however due to monogamy of entanglement there are non-trivial constraints amongst the geometric data for any pair of qubits. In some forth-coming work, this has been made precise in terms of  ``monogamy of steering'' relations \cite{Antony}, which in turn have provided a simple derivation of the Coffman-Kundu-Wootters inequality,  and elegant geometric connections between concurrence and maximal volume ellipsoids for a fixed center. 

Another direction of interest is motivated by a theorem of Petz, which states that a tripartite quantum state $\rho_{ABC}$, for which the conditional mutual information obeys $I(A:C|B) = 0$, can be decomposed as $\rho_{ABC} = \rho_{BC}^{1/2} \rho_B ^{-1/2} \rho_{AB} \rho_B^{-1/2}\rho_{BC}^{1/2}$. Note that the middle term on the right-hand side is precisely the canonical state $\tilde{\rho}_{AB}$, and so it would be of interest to formulate this relation in terms of ellipsoids, especially since the Petz relation admits generalization to the decomposition of multipartite systems in terms of bipartite entanglement. Such an analysis in terms of the geometry of steering could potentially be a valuable tool in the study of entanglement in multipartite mixed states. 

More broadly, our work on the relations between the various geometric data may also be an intuitive way to study the compatibility of bipartite reduced states with an overall multipartite quantum state (an instance of the ``quantum marginal problem'').
\end{document}